\documentclass[a4paper,UKenglish,cleveref, autoref, thm-restate]{lipics-v2021}
\usepackage{bbm}

\newcommand{\MC}{\text{MC}}
\newcommand{\MV}{\text{MV}}
\newcommand{\SFMP}{\text{SF-MP}}
\newcommand{\SFMV}{\text{SF-MV}}

\newcommand{\DFMP}{\text{DF-MP}}
\newcommand{\DFMV}{\text{DF-MV}}

\newcommand{\MP}{\text{MP}}

\DeclareMathOperator*{\E}{\mathbb{E}}

\DeclareMathOperator*{\argmax}{arg\,max}
\newcommand{\gs}{U} % ground set
\renewcommand{\ng}{\gamma} % number of groups
\usepackage{tikz}

\usepackage{booktabs}
\definecolor{dgreen}{rgb}{0,.6 ,0}
\definecolor{dcyan}{rgb}{0,.6 ,.6}

\hideLIPIcs

\title{Maximin Fairness in Combinatorial Optimization: Random Solutions and Insights through Max-Cut} %TODO Please add
\title{Expected Maximin Fairness in Combinatorial Optimization: Insights through Max-Cut}
\title{Expected Maximin Fairness in Combinatorial Optimization}
\title{Expected Maximin Fairness: Insights through Max-Cut}
\title{Expected Maximin Fairness in Max-Cut and other Combinatorial Optimization Problems}

\author{Jad Salem$^1$}{Mathematics Department, United States Naval Academy, United States\and \url{https://jadsalem.com/} }{jsalem@usna.edu}{https://orcid.org/0000-0003-4480-4601}{}%TODO mandatory, please use full name; only 1 author per \author macro; first two parameters are mandatory, other parameters can be empty. Please provide at least the name of the affiliation and the country. The full address is optional. Use additional curly braces to indicate the correct name splitting when the last name consists of multiple name parts.

\author{Reuben Tate\footnote{Corresponding author}}{CCS-3: Information Sciences, Los Alamos National Laboratory, United States \and \url{https://www.reubentate.com/}}{rtate@lanl.gov}{https://orcid.org/0000-0002-9170-8906}{}

\author{Stephan Eidenbenz}{CCS-3: Information Sciences, Los Alamos National Laboratory, United States}{eidenben@lanl.gov}{https://orcid.org/0000-0002-2628-1854}{}

\authorrunning{J. Salem, R. Tate, and S. Eidenbenz} %TODO mandatory. First: Use abbreviated first/middle names. Second (only in severe cases): Use first author plus 'et al.'

\Copyright{Jad Salem, Reuben Tate, and Stephan Eidenbenz} %TODO mandatory, please use full first names. LIPIcs license is "CC-BY";  http://creativecommons.org/licenses/by/3.0/

 %TODO mandatory; please add comma-separated list of keywords

\category{} %optional, e.g. invited paper

\relatedversion{} %optional, e.g. full version hosted on arXiv, HAL, or other respository/website
\nolinenumbers %uncomment to disable line numbering

\EventEditors{John Q. Open and Joan R. Access}
\EventNoEds{2}
\EventLongTitle{42nd Conference on Very Important Topics (CVIT 2016)}
\EventShortTitle{CVIT 2016}
\EventAcronym{CVIT}
\EventYear{2016}
\EventDate{December 24--27, 2016}
\EventLocation{Little Whinging, United Kingdom}
\EventLogo{}
\SeriesVolume{42}
\ArticleNo{23}
\begin{document}

\maketitle

\begin{abstract}
Maximin fairness is the ideal that the worst-off group (or individual) should be treated as well as possible. Literature on maximin fairness in various decision-making settings has grown in recent years, but theoretical results are sparse. In this paper, we explore the challenges inherent to maximin fairness in combinatorial optimization. We begin by showing that (1) optimal maximin-fair solutions are bounded by non-maximin-fair optimal solutions, and (2) stochastic maximin-fair solutions exceed their deterministic counterparts in expectation for a broad class of combinatorial optimization problems. In the remainder of the paper, we use the special case of Max-Cut to demonstrate challenges in defining and implementing maximin fairness. 
\end{abstract}

\newpage

\section{Introduction}
\label{sec:introduction}

Over the past couple decades, automated computer systems and algorithms have taken over much of the decision-making processes that are present in various companies, government agencies, and other organizations. Prior to this, decisions were often made without the aid of such algorithms, meaning that a decision maker's subjective opinions may potentially skew the decisions in a way that is unfair to certain groups or individuals. It is sometimes believed that automated computer algorithms (e.g., machine learning or AI-based algorithms) are ``objective'' and immune to such issues; however, without proper care, such systems may still potentially negatively impact certain groups or individuals due to biases in historical data, unintended positive feedback loops, or various incorrect assumptions \cite{de2019bias,kleinberg2018selection,dwork2012fairness,moss2012science,goodman2018amazon}.

In this work, we consider 
%a particular notion of 
group fairness in the context of combinatorial optimization. In the field of welfare economics or social choice theory, a social welfare function  assigns relative rankings of social states \cite{arrow2012social}; in particular, a cardinal social welfare function takes as input a collection of individual utilities and outputs a real number. Without considering fairness, it is common to consider the utilitarian social welfare function, which simply sums or averages all of the individual utilities. An alternative is to consider a \emph{maximin} or \emph{Rawlsian} (introduced by John Rawls) function which stipulates that the utility of the least well-off individual should be maximized \cite{rawls2017theory}. Given a simple notion of group utility (e.g. sum of individual utilities in the group), the Rawlsian social welfare function can be extended for the purposes of group fairness, i.e., the group utility of the worst-off group is optimized.

In general, it is often the case that \emph{any} decision or solution causes at least one individual and/or group to be negatively impacted. For this reason, it often beneficial to consider a \emph{distribution} of solutions; one can imagine that such distributions ``smooths out'' any unfairness that exists amongst individual solutions evenly across the groups. In practice, one can sample from such a distribution iteratively over time in the hopes of achieving some notion of ``fairness across time.'' We discuss this temporal aspect in the discussion (Section \ref{sec:discussion}). We will refer to the problem of choosing a deterministic solution to maximize a fairness objective as \emph{static fairness}, and the problem of choosing a distribution over solutions to maximize a fairness objective as \emph{dynamic fairness}.

% In practice, such distributions can be used with time-invariant utilities to obtain a sense of ``fairness-over-time"; we discuss this temporal aspect in the discussion (Section \ref{sec:discussion}). When it comes to using single deterministic solutions vs random solutions (e.g. from a distribution), we refer to these notions of fairness as \emph{static fairness} and \emph{dynamic fairness} respectively.

The purpose of this work is to build a theoretical foundation that considers the impact of dynamic fairness (compared to static fairness) in the context of combinatorial optimization problems in a Rawlsian framing. In addition to results for general combinatorial optimization problems, we also provide results and insights for the well-known Max-Cut problem which asks one to partition the vertices of a graph into two groups in a way that maximizes the number of edges crossing the partition. We show that one can put a utility function on the vertices or the edges in a way where solving Max-Cut is equivalent to optimizing an objective from a utilitarian perspective; as such, these edge or node utilities naturally induce a notion of ``fair Max-Cut.'' Although Max-Cut is not typically viewed as a resource allocation problem about which fairness issues are of concern, we choose this as our running example for a few reasons. First, the problem is well-studied yet non-trivial to solve (e.g., it is a quadratic binary optimization problem with no constraints); this makes it theoretically interesting to analyze. Second, it is believed that quantum computers have the potential to outperform classical computers on certain classes of combinatorial optimization problems, with much of the quantum optimization literature focusing on the Max-Cut problem \cite{guerreschi2019qaoa,tate2024guarantees,zhou2023qaoa,gaidai2024performance,wurtz2021maxcut,majumdar2021optimizing}. Since quantum optimization can efficiently produce distributions over solutions that would be intractable to produce under classical optimization, quantum computing has the potential to produce fairer dynamic solutions to this problem.

Moreover, considering \emph{distributions} of solutions (via dynamic fairness) is well-motivated since it is believed that quantum computers excel at generating distributions that are difficult for classical computers to produce \cite{bouland2019complexity}; we further discuss the potential impact of quantum computing on the fields of fairness and optimization in Section \ref{sec:discussion}. Third, we believe that the challenges arising from enforcing fairness in Max-Cut might shed light on fair versions of other graph optimization problems.

% \textcolor{red}{Jad--revise language}At the time of the writing of this manuscript, we are not aware of any concrete, real-life applications for ``fair Max-Cut"; however, we choose this as our running example for a few reasons. First, the problem is well-studied and relatively simple yet non-trivial to solve (e.g., it is a quadratic binary optimization with no constraints); this makes it appealing to analyze from a purely theoretical perspective. Secondly, it is believed that quantum computers have the potential to outperform classical computers on certain classes of combinatorial optimization problems \cite{farhi2014quantum} with much of the quantum optimization literature focusing on the Max-Cut problem \cite{guerreschi2019qaoa,tate2024guarantees,zhou2023qaoa,gaidai2024performance,wurtz2021maxcut,majumdar2021optimizing}; we further discuss the potential impact of quantum computing on the fields of fairness and optimization in the discussion (Section \ref{sec:discussion}).

The paper is organized as follows. In Section \ref{sec:related-work}, we discuss related work that considers similar notions of fairness. In Section \ref{sec:maximinFairnessInCombOpt}, we provide general results that are slightly more generalized than the ``utility-based group Rawlsian fairness over distributions'' framework described earlier. In Section \ref{sec:maxcut}, we discuss the consequences of Section \ref{sec:maximinFairnessInCombOpt} in regards to Max-Cut while also providing additional results that are specific to the Max-Cut problem. Finally, we provide a discussion and conclude in Section \ref{sec:discussion}.

\section{Related Work}
\label{sec:related-work}

In this section, we position our work with respect to the literature on maximin fairness, static solutions to combinatorial optimization problems, and Max-Cut.

% While there is excitement about maximin fairness in combinatorial optimization, the relevant theoretical literature has been slow to grow. As our work... To better situate our work, we discuss advances in maximin fairness and Max-Cut.

\subparagraph*{Maximin fairness in combinatorial optimization.} 
In this paper, we theoretically analyze maximin fairness in combinatorial optimization problems and refine our analysis for Max-Cut problems. In a similar vein, Garcia-Soriano and Bonchi consider maximin fairness in constrained ranking \cite{garcia2021maxmin}. In particular, their objective is to choose a ranking which maximizes the minimum utility incurred by an individual. The authors find a maximin-optimal distribution over rankings using linear programming and a maximin-optimal deterministic ranking using a greedy algorithm. Similarly, we will show how to find optimal distributions over cuts for two different maximin objectives in Max-Cut. Filippi et al. consider a Rawls-esque objective in facility location which seeks to minimize the average cost paid by the worst-off customers \cite{filippi2021single}. Unlike Garcia-Soriano and Bonchi, they formulate their problem with two objectives (one of which is the fairness objective) and provide a method for finding Pareto-optimal solutions. 
Despite these few notable exceptions, the theoretical literature on maximin fairness in combinatorial optimization has been slow to grow and theoretical literatire on maximin fairness remains sparse.

%The theoretical literature on maximin fairness in combinatorial optimization, however, remains sparse. 
% Stephan writes: integrated previosu two sentences

Instead, empirical work analyzing maximin-fair solutions has become more common. %quantifying the degree of maximin fairness acheived has become more common. 
For example, Blanco and G{\' a}zquez %\textcolor{red}{citations!}  
consider ordered utility aggregators as objectives in the coverage facility location problem, and certain maximin fair objectives are special cases of these aggregators \cite{blanco2023fairness}. The authors empirically compare these different aggregators by finding optimal solutions and comparing their Gini indices. Our work, on the other hand, seeks to theoretically analyze the impact of maximin objectives on combinatorial optimization problems, and where possible, to describe the relationship between optimal solutions and the structure of the underlying problem.

% \begin{itemize}
%     \item Fair facility location \cite{filippi2021single}
%     \begin{itemize}
%         \item They consider ``conditional $\beta$-means'' as a fairness measure, which I would argue can be motivated by the Rawlsian principle. The conditional $\beta$-mean is the ``average cost paid by the $100\beta$\% of the demand that pays the largest unit costs to reach the assigned facility.''
%         \item They frame the problem as bi-objective (one to minimize cost, one to maximize fairness)
%     \end{itemize}

%     \item Fair coverage facility location \cite{blanco2023fairness}
%     \begin{itemize}
%         \item They consider the coverage facility problem with general ordered utility aggregators as fairness measures. Max-min fairness is a special case.
%         \item They frame the problem as a maximization problem, where the objective is fairness (it could the the average utility, the average of the top k utilities, etc).
%     \end{itemize}
% \end{itemize}

\subparagraph*{Distributions over solutions in combinatorial optimization.} 
Randomization can pose several advantages over deterministic solutions in optimization, from increasing expected performance \cite{garcia2021maxmin}, to hedging against adversarial behavior \cite{salem2024secretary}, to avoiding routinely disadvantaging individuals when used iteratively over time \cite{lodi2023framework,lodi2024fairness}. The authors in \cite{garcia2021maxmin}, for example, find maximin-optimal distributions over rankings (i.e., stochastic solutions) and maximin-optimal deterministic rankings (i.e., deterministic solutions) and empirically show performance gaps between the two. 

In cases where decisions are made one-by-one over time,\footnote{Here, we are referring to time-invariant online optimization problems without feedback (i.e., not online learning problems), since in the presence of feedback, one would update their decision-making strategy as feedback is accrued.} one strategy is to find a distribution over solutions and sample from it iteratively. In \cite{lodi2023framework} and \cite{lodi2024fairness}, for example, the authors discuss the problem of minimizing a function (e.g., a cost or unfairness function) of a sequence of decisions, noting that the average performance over a sequence of decisions can be better than the performance at any single iteration. In a similar spirit, we will probe the performance differences stemming from optimal static decisions and optimal dynamic decisions in the context of a maximin-fair Max-Cut problem.

\subparagraph*{Max-Cut.} Max-Cut is the problem of choosing a subset $S$ of vertices in a graph which maximizes the number of edges with exactly one endpoint in $S$. Max-Cut is an NP-hard optimization problem \cite{karp2010reducibility} for which the best-known polynomial-time approximation ratio (due to Goemans and Williamson) is approximately 0.878 \cite{goemans1995improved}. However, the Max-Cut problem has recently gained traction in the field of quantum computing, as quantum algorithms have seen empirically better approximation ratios than the Goemans-Williamson algorithm on certain classes graphs \cite{crooks2018performance}. We discuss intersections of quantum computing and fairness further in Section~\ref{sec:discussion}.

\section{Maximin Fairness in Combinatorial Optimization}
\label{sec:maximinFairnessInCombOpt}

In this section, we first define our notions of maximin fairness in combinatorial optimization (Section~\ref{sec:notation}; then, we provide examples to illustrate the breadth of our model (Section~\ref{subsec:examples}); then, we probe the relationships between the various optimization problems defined in Section~\ref{sec:notation} (Section~\ref{sec:generalBounds}); finally, we discuss a simple randomized algorithm that hedges across the $\ng$ groups (Section~\ref{sec:separate-a}).

\subsection{Notation and Definitions}
\label{sec:notation}

We consider a population of individuals $\gs$; we sometimes refer to $U$ as the \emph{ground set}. Next, we assume that we are given a partition $\Gamma$ of $\gs$ into $\gamma$ non-empty groups: $$\Gamma = \{\gs_1, \dots, \gs_\gamma\} \text{ with } \gs = \gs_1 \sqcup \gs_2 \cdots \sqcup \gs_\gamma$$ where $\sqcup$ denotes a \emph{disjoint} union. In practice, these groups could be determined by legally protected attributes, various departments in an organization, or any number of other criteria. Note that it is possible for each group to contain exactly one individual, i.e., $|\gs_i|=1$ for $i=1,\dots,\gamma$ with $\gamma = |\gs|$; in such a case, the definitions presented here reduce to an \emph{individual} notion of fairness.

Next, we are given a set of decisions $\mathcal{X}$ with each decision potentially providing varying amounts of utility to various subsets of the population. More precisely, each decision $x\in \mathcal{X}$ has an associated set function 
\[
f_x : \mathcal{P}(\gs) \to \mathbb{R} ~~\mbox{for $x \in \mathcal{X}$}
\]

% Suppose we have a finite set $\mathcal{X}$ of decisions and a finite set $\gs$ of individuals. The set $\gs$ partitioned as follows:
% \[
% \Gamma = \{\gs_1,\ldots,\gs_{\ng}\}, ~~~~\gs = \gs_1 \sqcup \gs_2 \sqcup \cdots \sqcup \gs_\ng. %~~~~\gs_i = \{u^i_j : j=1,\ldots,|\gs_i|\}.
% \]
%Now suppose we have the following set functions

defined on the power set of $\gs$; here, $f_x(A)$ is the utility incurred by the subgroup of people $A$ from decision $x$. For many problems of interest, it is often the case that the utility that a subgroup $A$ incurs is simply the sum of individual utilities, i.e., $f_x(A) = \sum_{a\in A} f_x(\{a\})$; however, this is not necessarily the case.

From a utilitarian perspective, one would want to find a decision $x\in X$ that maximizes either the total utility of the population, $f_x(U)$ or the per-capita utility of the whole population, i.e., $f_x(\gs)/|\gs|$. We refer to such problems as the Maximum Value (\MV) and Maximum Proportion (\MP) problem respectively, defined below. We remark that both of these problems are equivalent in the sense that the $x$ that is optimal for the Maximum Value problem will also be optimal for the Maximum Proportion problem (and vice-versa).
\begin{align*}
    \MV &= \max_{x \in \mathcal{X}} f_x(\gs)  &\mbox{maximum value}\\
    \MP &= \max_{x \in \mathcal{X}} \frac{f_x(\gs)}{|\gs|} &\mbox{maximum proportion}
\end{align*}

We next define the Rawlsian \cite{rawls2017theory} framework discussed in the introduction. To this end, one needs a way to compare fairness between groups. One way to accomplish this is by considering, for any decision $x\in X$, the \emph{per-capita utility} within each group, i.e., $f_x(\gs_i)/|\gs_i|$. With this per-capita measure, for any decision $x\in X$, the worst-off group is given by $\min_{i\in [\gamma]} f_x(\gs_i)/|\gs_i|$. The Rawlsian framework would then stipulate that we should pick a solution from $\mathcal{X}$ in a way that makes the worst-off group as best-off as possible; we define such an optimization problem as the \emph{Static Fair Maximum Proportion} (\SFMP) problem defined below: 
%Due to outliers, it may be difficult to achieve fairness for every individual in a population; instead, the per-capita utility gives us a rough proxy of the utilities felt by the individuals within each group.
\begin{align*}
    \SFMP &= \max_{x \in \mathcal{X}} \min_{i \in [\ng]} \frac{f_x(\gs_i)}{|\gs_i|} &\mbox{static fair maximum proportion}
\end{align*}

As we will later see in this paper, it is sometimes the case that for \emph{any} decision $x\in X$, that there is receives little to no utility, effectively bringing down the $\SFMP$. However, this can be rectified by sampling a decision $x$ from a {\it distribution} $D$; that is, one could choose decision $x$ with probability $D(x)$, for all $x \in \mathcal{X}$. In this way, even though a group $\gs_i$ might incur low utility from {\it some} decisions $x$ with non-zero support in $D$, the %However, this can potentially be rectified by considering a \emph{distribution} $D$ of decisions \textcolor{red}{STephan writes: define what you mean by distribution, eg, assigns a probability to each possible decision x}: a group $\gs_i$ might incur low utility from \emph{some} decisions $x$ with non-zero support in $D$, but its 
\emph{expected}  per-capita utility $\E_{x\sim D} \frac{f_x(U_i)}{|U_i|}$ might still be acceptable. Now, for any choice of a \emph{distribution} $D$, similar to the \SFMP{} problem, we can define the worst-off group to be $\min_{i\in[\gamma]} \E_{x\sim D} \frac{f_x(U_i)}{|U_i|}$. The Rawlsian framework would then again stipulate that we make this quantity as large as possible, thus motivating the Dynamic Fair Maximum Proportion (\DFMP) below:
\begin{align*}
    \DFMP &= \max_{D \in \mathcal{D}} \min_{i \in [\ng]} \E_{x \sim D} \frac{f_x(\gs_i)}{|\gs_i|} &\mbox{dynamic fair maximum proportion}
\end{align*}
where $\mathcal{D}$ is the set of distributions over the decision space $\mathcal{X}$. %\textcolor{red}{Stephan writes: drop the big delta OR formally define it, but really dropping it is better I think. Just say fancy D is the set of distributions of decision space fancy X.}

It may also be convenient to express some objective functions as {\it unnormalized} utility functions (i.e., the utility to group $\gs_i$ would be $f_x(\gs_i)$ instead of $f_x(\gs_i)/|\gs_i|$). This distinction is meaningless for general functions $f_x$, since the normalization factor of $1/|\gs_i|$ could be incorporated into $f_x$. That said, if one wishes to normalize utility in a different way, then the unnormalized objective is more natural. To that end, we introduce the Static Fair Maximum Value ($\SFMV$) and the Dynamic Fair Maximum Value ($\DFMV$), defined below.

% \textcolor{red}{Jad} For the sake of completion, we also consider variants of the $\SFMP$ and $\DFMP$ where instead of comparing groups by per-capita utility, we compare them via their total utility. We refer to these variants as the Static Fair Maximum Value (\SFMV) and the Dynamic Fair Maximum Value (\DFMV) problem, defined below. Although these optimization problems are not as well-motivated as the others, we do present some results for these problems in Section \ref{sec:generalBounds} which may potentially be of interest to some readers.
\begin{align*}
    \SFMV &= \max_{x \in \mathcal{X}} \min_{i \in [\ng]} f_x(\gs_i) &\mbox{static fair maximum value} \\
    \DFMV &= \max_{D \in \mathcal{D}} \min_{i \in [\ng]} \E_{x \sim D} f_x(\gs_i) &\mbox{dynamic fair maximum value}
\end{align*}

% In this paper, we will consider both unconstrained objectives and maximin fair objectives. We also find it helpful to distinguish between  (1) static solutions $x \in \mathcal{X}$ and distributions over solutions, and (2) group-normalized utilities and unnormalized utilities. In particular, consider the following optimization problems: \textcolor{red}{present differently?}
% \begin{align*}
%     \MV &= \max_{x \in \mathcal{X}} f_x(\gs)  &\mbox{maximum value}\\
%     \SFMV &= \max_{x \in \mathcal{X}} \min_{i \in [\ng]} f_x(\gs_i) &\mbox{static fair maximum value} \\
%     \DFMV &= \max_{D \in \mathcal{D}} \min_{i \in [\ng]} \E_{x \sim D} f_x(\gs_i) &\mbox{dynamic fair maximum value}\\
%     \MP &= \max_{x \in \mathcal{X}} \frac{f_x(\gs)}{|\gs|} &\mbox{maximum proportion}\\
%     \SFMP &= \max_{x \in \mathcal{X}} \min_{i \in [\ng]} \frac{f_x(\gs_i)}{|\gs_i|} &\mbox{static fair maximum proportion}\\
%     \DFMP &= \max_{D \in \mathcal{D}} \min_{i \in [\ng]} \E_{x \sim D} \frac{f_x(\gs_i)}{|\gs_i|} &\mbox{dynamic fair maximum proportion}\\
% \end{align*}
 In Section \ref{subsec:examples} below, we provide examples of combinatorial optimization problems that fall into the general framework described above. 

\subsection{Examples}
\label{subsec:examples}

The setting above can be applied to many combinatorial optimization problems, including the following:
\begin{itemize}
    \item {\bf Maximum Independent Set (MIS).} The goal in MIS is to choose an independent set of maximum size. In this case, given an independent set $x$ and a subset $A$ of vertices, one could define $f_x(A) = |A \cap x|$. In this case, given a partition $\Gamma = \{V_1,\ldots,V_\gamma\}$ of the vertex set, $f_x(V_i)$ is the number of vertices in $V_i$ that are in the independent set $x$. The static fair maximum value ($\SFMP$) would therefore maximize the per-capita group utility $f_x(A)/|A|$ for the worst-off group in $\Gamma$. 
    
    $~~~~$To highlight the difference between Rawlsian fairness and other notions of fairness, consider the trivial independent set $x' = \varnothing$. Since all groups incur equal utility (in particular, a utility of 0), $x'$ would satisfy the well-studied demographic parity constraint \cite{feldman2015certifying}. However, there exists an independent set $x''$ that maximizes the Rawlsian objective $\min_{i \in [\ng]} |V_i \cap x|/|V_i|$ that would yield positive utility for at least one group. Thus, while utilities under $x''$ may vary across groups, they weakly dominate the utilities under $x'$.

    % In some sense, this is the ``most'' fair since every group would incur equal amounts of utility, e.g., a utility of $0$. Meanwhile, for the independent set $x''$ that maximizes $\min_{i \in [\ng]} |V_i \cap x|/|V_i|$, while $x''$ may not give equal amounts of utility to each group, all groups (including the worst-off group) would prefer the solution $x''$, since no group would incur less utility from $x''$ than $x'$.
    
    % each group would likely (depending on the problem instance) incur a non-zero utility.

    % \textcolor{red}{say something about what the optimization problems mean}
    
    \item {\bf Longest Path Problem (LPP).} The goal in LPP is to select a path in an edge-weighted graph of maximum weight. Given a partition $V_1,\ldots,V_m$ of the vertex set $V$, one may wish to avoid devoting more resources to one group of nodes than to the others. In this case, one could define
    \[
    f_x(A) = \frac{1}{2}\sum_{v \in A} w(\{uv : u \in N(v)\} \cap x),
    \]
    for any $A \subseteq V$, where $w(\{uv : u \in N(v)\} \cap x)$ is the sum of the weights of edges in $x$ that are incident to $v$. 

    \item {\bf Maximum Clique (MC)}. Given a graph $G=(V,E)$, the goal in MC is to select a clique of maximum size. Suppose that a group $A \subseteq V$ gains utility from a clique $x$ dependent on the number of edges of $G[A]$ that are present in $x$. For instance, one could have $f_x(A) = |E(G[A]) \cap x|$. In this case, the optimal solution to the MC problem is $\argmax_{x \in C} f_x(V)$, where $C$ is the set of cliques in $G$. Also note that $f_x$ is nonnegative and superadditive.

    % \item Facility location with coal.... cost is $f_x(A) \geq 0$
    % \item MSP, traffic volume between pairs, overall load should be even (nodes all want to communicate, node incurs negative utility when it transmits message that is not its own; want to even out expected utility).  \textcolor{red}{Jad}
\end{itemize}
Note that in each of these examples, the condition of Corollary~\ref{cor:general-bound} is satisfied.

\subsection{General Bounds}
\label{sec:generalBounds}

Many of the results in this work require our decision set functions $\{f_x\}_{x\in\mathcal{X}}$ to have certain properties, such as non-negativity. We define some properties of interest below.

\begin{definition}
    A set function $f : \mathcal{P}(\gs) \to \mathbb{R}$ is \emph{superadditive} if $f(A \cup B) \geq f(A) + f(B)$ for any $A,B \subseteq \gs$ with $A \cap B = \emptyset$.
\end{definition}

A related concept is that of supermodularity, which captures the notion of increasing returns.

\begin{definition}
    A set function $f : \mathcal{P}(\gs) \to \mathbb{R}$ is \emph{supermodular} if $f(A \cup B) + f(A \cap B) \geq f(A) + f(B)$ for any $A,B \subseteq \gs$.
\end{definition}

We begin by observing that nested subproblems experience monotonicity. In particular, in the next theorem, we will bound the $\DFMP$ of a problem in terms of the maximum proportion of a subproblem, and similarly for the unnormalized objectives.  The full proof of the theorem below can be found in Appendix~\ref{sec:proofs_section_3}.

\begin{theorem}[monotonicity of subproblems] \label{thm:subproblem-bounds}
Consider 
\begin{enumerate}
    \item an optimization problem with a ground set $\gs$ with partition $\Gamma = \{\gs_1,\ldots,\gs_{\ng}\}$ and nonnegative utility functions $f_x : \mathcal{P}(\gs) \to \mathbb{R}$, for $x \in \mathcal{X}$, and
    \item an optimization problem with a ground set $\gs'$ with partition $\Gamma' = \{\gs'_1,\ldots,\gs'_{\ng'}\} \subseteq \Gamma$ and superadditive utility functions $f_x' : \mathcal{P}(\gs') \to \mathbb{R}$ for $x \in \mathcal{X}$.%\textcolor{red}{Stephan writes: define superadditive}
\end{enumerate} 
Let $\DFMV$ and $\DFMP$ be the dynamic fair maximum value and proportion of the first problem, and let $\MV'$ and $\MP'$ be the maximum value and proportion of the second. Then
\begin{align*}
\DFMV &\leq \MV' + \sum_{i=1}^{\ng'} \delta(\gs'_i) \mbox{, and}\\
\DFMP &\leq \MP' + \frac{\sum_{i=1}^{\ng'} \delta(\gs'_i)}{\sum_{i=1}^{\ng'} |\gs'_i|},
\end{align*}
where $f_x'(\gs'_i) + \delta(\gs'_i) \geq f_x(\gs'_i)$ for $i =1,\ldots,\ng'$.
\end{theorem}

\begin{proof}[Proof sketch]
We give a sketch of the proof of the second inequality here. To that end, let $D \in \argmax_{D \in \mathcal{D}} \min_{i \in [\ng]} \E_{x \sim D} f_x(\gs_i)/|\gs_i|$. The main steps of the analysis are to bound the maximum proportion of the subproblem by the sum of group objectives, and then to bound this sum of group objectives by the Rawlsian objective of the subproblem, as outlined below.
% We prove the inequality involving $\DFMP$; the proof for the inequality involving $\DFMV$ can be found in Appendix \ref{sec:proofs_section_3}. To that end, let $D \in \argmax_{D \in \mathcal{D}} \min_{i \in [\ng]} \E_{x \sim D} f_x(\gs_i)/|\gs_i|$. Then
\begin{align*}
    \MP' + \frac{\sum_{i=1}^{\ng'} \delta(\gs'_i)}{\sum_{i=1}^{\ng'} |\gs'_i|} 
    &\geq \E_{x \sim D} \frac{\sum_{i=1}^{\ng'} f_x'(\gs'_i)}{\sum_{i=1}^{\ng'} |\gs'_i|} + \frac{\sum_{i=1}^{\ng'} \delta(\gs'_i)}{\sum_{i=1}^{\ng'} |\gs'_i|} &\mbox{by superadditivity} \\
    &\geq \min_{i \in [\ng']} \frac{\E_{x \sim D} f_x(\gs_i')}{|\gs'_i|} &\mbox{by nonnegativity} \\
    &\geq \DFMP &\mbox{since $\displaystyle D \in \argmax_{D \in \mathcal{D}} \min_{i \in [\ng]} \E_{x \sim D} \frac{f_x(\gs_i)}{|\gs_i|}$.}
\end{align*}
The analysis of the first inequality is similar, except that a different justification is required for bounding the sum of group utilities by the $\DFMV$.
\end{proof}

As a consequence, one can bound the maximum value and the maximum proportion by their dynamic fair counterparts. We summarize the relationships between the optimization problems presented above in the corollary below.

\begin{corollary} \label{cor:general-bound}
Suppose $f_x$ is nonnegative and superadditive for all $x \in \mathcal{X}$. Then
\begin{align*}
&\SFMV \overset{(a)}{\leq} \DFMV \overset{(b)}{\leq} \MV, ~~\mbox{and}\\ 
&\SFMP \overset{(c)}{\leq} \DFMP \overset{(d)}{\leq} \MP.
\end{align*}
\end{corollary}

\begin{proof}
First, we prove (a) and (c). To that end, let $\mathcal{D}'$ be the set of Dirac distributions over $[N]$ (so, each distribution in $\mathcal{D}'$ is the indicator of a singleton). Then, since $\mathcal{D}' \subset \mathcal{D}$,
\begin{align*}
    \SFMV &= \max_{D \in \mathcal{D}'} \min_{i \in [\ng]} \E_{x \sim D}   f_x(\gs_i) \leq \max_{D \in \mathcal{D}} \min_{i \in [\ng]} \E_{x \sim D}   f_x(\gs_i) = \DFMV, ~~\mbox{and} \\
    \SFMP &= \max_{D \in \mathcal{D}'} \min_{i \in [\ng]} \E_{x \sim D}   \frac{f_x(\gs_i)}{|\gs_i|} \leq \max_{D \in \mathcal{D}} \min_{i \in [\ng]} \E_{x \sim D}   \frac{f_x(\gs_i)}{|\gs_i|} = \DFMP.
\end{align*}
Next, observe that (b) and (d) follow from Theorem~\ref{thm:subproblem-bounds} by setting $\Gamma' = \Gamma$ and $f'_x = f_x$ for all $x \in \mathcal{X}$.
% Next, we show (b). Letting $D \in \argmax_{D \in \mathcal{D}} \min_{i \in [\ng]} \E_{x \sim D} f_x(\gs_i)$, we have that
% \begin{align*}
%     \MV &= \max_{x \in \mathcal{X}} f_x(\gs) \\
%     &\geq \E_{x \sim D} f_x(\gs) \\
%     &\geq \E_{x \sim D} \sum_{i \in [\ng]} f_x(\gs_i) &\mbox{by superadditivity} \\
%     &= \sum_{i \in [\ng]} \E_{x \sim D} f_x(\gs_i) \\
%     &\geq \min_{i \in [\ng]} \E_{x \sim D} f_x(\gs_i) &\mbox{by nonnegativity}\\
%     &= \DFMV &\mbox{since $D \in \argmax_{D \in \mathcal{D}} \min_{i \in [\ng]} \E_{x \sim D} f_x(\gs_i)$}
% \end{align*}
% To show (d), let $D \in \argmax_{D \in \mathcal{D}} \min_{i \in [\ng]} \E_{x \sim D} f_x(\gs_i)/|\gs_i|$. Then
% \begin{align*}
%     \MP &= \max_{x \in \mathcal{X}} \frac{f_x(\gs)}{|\gs|} \\
%     &\geq \E_{x \sim D} \frac{f_x(\gs)}{|\gs|} \\
%     &\geq \E_{x \sim D} \frac{\sum_{i \in [\ng]} f_x(\gs_i)}{\sum_{i \in [\ng]}|\gs_i|} &\mbox{by superadditivity} \\
%     &= \frac{\sum_{i \in [\ng]} \E_{x \sim D} f_x(E_i)}{\sum_{i \in [m]}|E_i|} \\
%     &\geq \min_{i \in [\ng]} \E_{x \sim D} \frac{f_x(\gs_i)}{|\gs_i|} &\mbox{by nonnegativity} \\
%     &= \DFMP &\mbox{since $D \in \argmax_{D \in \mathcal{D}} \min_{i \in [\ng]} \E_{x \sim D} f_x(\gs_i)/|\gs_i|$.}
% \end{align*}
% This proves the claim.
\end{proof}

\begin{observation}
If $f_x$ is supermodular, nonnegative, and satisfies $f_x(\emptyset)=0$ for all $x \in \mathcal{X}$, then the condition of Corollary~\ref{cor:general-bound} holds.
\end{observation}

\begin{observation}
The notion of average utility used in the definitions of $\DFMP$ and $\SFMP$ normalizes the total utility $f_x(\gs_i)$ by the $|\gs_i|$. However, in some cases, it may be of interest to normalize by some other value, $g(\gs_i)$. For example, if the elements of $\gs$ are weighted, $g(\gs_i)$ could be of sum of weights of elements of $\gs_i$. In this more general setting, Corollary~\ref{cor:general-bound} holds if $g$ is subadditive and nonnegative.
\end{observation}

% \begin{observation}
%     The above optimization problems can be seen as an implementation of the Rawls difference principle to combinatorial {\it maximization} problems. However, the same ideas can be implemented---and the same bounds can be derived---for combinatorial minimization problems. See Appendix~\ref{app:minimization} for details.
% \end{observation}

% \textcolor{red}{Obs: minimax works with reversed inequalities, proofs are in appendix--Jad}

% We end this section by observing that these optimization problems have a monotonic structure due to their maximin objectives. Under reasonable assumptions, excluding groups only...

\subsection{Separate-$\mathcal{A}$-Solve Algorithm}
\label{sec:separate-a}

Finding the optimal distribution $D^*$ associated with a problem's \DFMP{} value is typically going to be intractable for most problems of interest since the associated decision set $\mathcal{X}$ is already large, let alone the set of all possible distributions over $\mathcal{X}$. Thus, in practice, one would have to run an algorithm in the hopes of finding a distribution of solutions with a ``good enough" amount of dynamic fairness.

The best algorithms for maximizing dynamic fairness should  exploit the structure of the specific underlying optimization problem at hand. Nonetheless, we present a very basic algorithm that can be used as a baseline algorithm for future algorithms to compare against.

For this baseline algorithm, we assume we are already given a deterministic algorithm $\mathcal{A}$ that takes as input, a subset $\gs'$ of individuals from $\gs$, and returns as output, an decision $x \in \mathcal{X}$, often with the goal of making $f_x(\gs')$ as large as possible. Using $\mathcal{A}$, we develop a meta-algorithm called the \emph{Separate-$\mathcal{A}$-Solve algorithm} which, as the name suggests, creates a distribution of solutions by considering each group separately. More specifically, the Separate-$\mathcal{A}$-Solve algorithm works as follows: if $\gs$ is partitioned as $\Gamma = \{\gs_1, \dots, \gs_\gamma\}$, then $\mathcal{A}$ is run on each group $\gs_1$ through $\gs_\gamma$ to obtain decisions $x_1, \dots, x_\gamma$; the Separate-$\mathcal{A}$-Solve algorithm then simply returns each solution $x_i$ with uniform probability (i.e. with probability $1/\gamma$).

In Theorem \ref{thm:separate_A_alg} below, we show that as long as $\mathcal{A}$ is able to return a group utility that is at least some fraction of the group size for every group, then a lower bound on the dynamic fairness achieved by the corresponding Separate-$\mathcal{A}$-Solve algorithm can be obtained.

\begin{theorem}
\label{thm:separate_A_alg}
  Let $\gs$ be the ground set of an optimization problem with partition $\Gamma = \{\gs_1, \dots, \gs_\gamma\}$ and non-negative utility functions $f_x: \mathcal{P}(\gs) \to \mathbb{R}$ for each $x\in \mathcal{X}$. Let $\alpha \geq 0$ and let $\mathcal{A}$ be an algorithm that, on input of $i\in[\gamma]$, returns a solution $x_i \in \mathcal{X}$ %on group $\gs_i$ 
  with the property that $f_{x_i}(\gs_i) \geq \alpha |\gs_i|$. In this case, the {\it Separate-$\mathcal{A}$-Solve} algorithm attains the following dynamic fair bound: $\min_{i \in [\ng]} \E_{x \sim D} \frac{f_x(\gs_i)}{|\gs_i|} \geq \frac{\alpha}{\gamma}$.%Let $D$ be the corresponding distribution of the Separate-$\mathcal{A}$-Solve algorithm described above. Then the dynamic fairness with $D$, i.e., $\min_{i \in [\ng]} \E_{x \sim D} \frac{f_x(\gs_i)}{|\gs_i|}$, is at least $\frac{\alpha}{\gamma}$.
\end{theorem}
\begin{proof}
    Let $x_1, \dots, x_m$ be the decisions obtained from the Separate-$\mathcal{A}$-Solve algorithm. By the nature of the  Separate-$\mathcal{A}$-Solve algorithm, we have that  $f_{x_i}(\gs_i) \geq \alpha |\gs_i|$. Thus, we have that, for each $i=1,\dots,\gamma$, that
    $$\E_{x \sim D} \frac{f_{x}(\gs_i)}{|\gs_i|}  = \sum_{j=1}^\gamma \Pr(x = x_j)  \frac{f_{x_j}(\gs_i)}{|\gs_i|} = \sum_{j=1}^\gamma \frac{1}{\gamma}  \frac{f_{x_j}(\gs_i)}{|\gs_i|} \geq \frac{1}{\gamma}  \frac{f_{x_j}(\gs_j)}{|\gs_j|} \geq \frac{1}{\gamma}  \frac{\alpha |\gs_j|}{|\gs_j|} = \frac{\alpha}{\gamma},$$
    where the first inequality above is obtained by just considering the $j=i$ addend of the sum and the fact that all the $f_x$'s are non-negative. Since the above calculations hold for every $i\in[\gamma]$, then $\min_{i \in [\ng]} \E_{x \sim D} \frac{f_x(\gs_i)}{|\gs_i|} \geq \frac{\alpha}{\gamma}$.
\end{proof}

\begin{observation}
We note that Theorem \ref{thm:separate_A_alg} can be easily modified to the case where $\mathcal{A}$ is a \emph{randomized} algorithm instead of a deterministic algorithm as previously assumed.
\end{observation}

% It should be noted that Theorem \ref{thm:separate_A_alg} can be easily modified in the case where $\mathcal{A}$ is a \emph{randomized} algorithm instead of a deterministic algorithm as previously assumed. The details of the modification will depend on when the randomization occurs for the Separate-$\mathcal{A}$-solve algorithm, e.g., one could run $\mathcal{A}$ exactly once for each group to (randomly) obtain decisions $x_1, \dots, x_\gamma$, then, construct the distribution $D$ that simply returns each $x_i$ with probability $1/\gamma$. Alternatively, one could imagine obtaining samples by first choosing an integer $i$ uniformly from $[n]$ and then randomly returning the (random) solution obtained by running $\mathcal{A}$ on $\gs_i$.

\section{Max-Cut}
\label{sec:maxcut}
The Max-Cut problem is as follows: given a graph $G=(V,E)$, find a cut, i.e. a partition of the vertices into two disjoint groups, such that the number of edges between the two groups is maximized. The Max-Cut objective can mathematically be written as \begin{equation}\label{eqn:edgeCrossCutIndicator}\MC(G) = \max_{S \subseteq V} \sum_{e \in E} X_e(S),\end{equation} where, for any edge $(u,v) \in E$, $X_e := X_e(S) =  \mathbf{1}[e \text{ crosses the cut}]$.
% \begin{equation}
% \label{eqn:edgeCrossCutIndicator}
% X_e := X_e(S) =  \mathbf{1}[e \text{ crosses the cut}] = \mathbf{1}[(u \in S \wedge v \notin S) \vee (u \notin S \wedge v \in S)],\end{equation}
% where $\mathbf{1}[\ \cdot \ ]$ is an indicator function that maps to 1 if the inside proposition is true and 0 otherwise. Given $S \subseteq V$, we will sometimes interpret $S$ as a bitstring, i.e., we define $S_v = \mathbf{1}[v \in S]$.

As the Max-Cut objective is typically written as a sum over the edges of the graph, it is no surprise that the Max-Cut problem can be viewed as a problem of maximizing the average utility of individuals, where each edge corresponds to an individual. However, in several real-world applications (e.g., if the graph is a social network graph), it is more common for the \emph{nodes} to represent individuals; we will show that Max-Cut can also be written as a the average of utilities over the nodes. In Sections~\ref{sec:edgeUtilities}-\ref{sec:nodeUtilities}, we consider maximin-fair Max-Cut through the notions of edge and node utilities, respectively, and prove tight bounds on the optimization problems defined in Section~\ref{sec:notation}. See Table~\ref{tab:gaps} for selected results.

% the following subsections (Section \ref{sec:edgeUtilities} and Section \ref{sec:nodeUtilities}), we consider maximin-fair Max-Cut through the notions of edge and node utilities, respectively, and prove tight bounds on the optimization problems defined in Section~\ref{sec:notation} (see Table~\ref{tab:gaps}).

\begin{table}[]
    \centering
    \bgroup
    \def\arraystretch{1.3}
    \begin{tabular}{l||cc|cc} \bottomrule
     & \multicolumn{2}{c}{\textbf{Edge Utilities}} & \multicolumn{2}{c}{\textbf{Node Utilities}} \\
     & Bounds & Reference & Bounds & Reference \\ \hline 
     $\MP(G)-\DFMP(G,\Gamma)$ & [0, 1/2] & Prop. \ref{prop:mp-dfmp-edge}  & [0,1] & Prop. \ref{prop:mp-dfmp-node} \\
     $\DFMP(G,\Gamma) - \SFMP(G,\Gamma)$ & [0,1] & Prop. \ref{prop:dfmp-sfmp-edge} & [0, 1/2] & Prop. \ref{prop:dfmp-minus-sfmp-node} \\ \toprule
    \end{tabular}
    \egroup
    \caption{Tight bounds on the gaps between the optimization problems defined in Section~\ref{sec:notation} for the Max-Cut problems with edge utilities and with node utilities.}
    \label{tab:gaps}
\end{table}

% {
% \begin{center}
%     \bgroup
%     \def\arraystretch{1.3}
%     \begin{tabular}{l||cc|cc} \bottomrule
%      & \multicolumn{2}{c}{\textbf{Edge Utilities}} & \multicolumn{2}{c}{\textbf{Node Utilities}} \\
%      & Bounds & Reference & Bounds & Reference \\ \hline 
%      $\MP(G)-\DFMP(G,\Gamma)$ & [0, 1/2] & Prop. \ref{prop:mp-dfmp-edge}  & [0,1] & Prop. \ref{prop:mp-dfmp-node} \\
%      $\DFMP(G,\Gamma) - \SFMP(G,\Gamma)$ & [0,1] & Prop. \ref{prop:dfmp-sfmp-edge} & [0, 1/2] & Prop. \ref{prop:dfmp-minus-sfmp-node} \\ \toprule
%     \end{tabular}
%     \egroup
    
% \end{center}
% }

\subsection{Edge Utilities}
\label{sec:edgeUtilities}
We now express fair Max-Cut with edge utilities using the notation in Section \ref{sec:generalBounds}. The set of decisions is simply the possible choices of cuts; since each cut is determined by one side of the cut, we can write $\mathcal{X} = \mathcal{P}(V)$. Our set of individuals $\gs$ (as described in Section \ref{sec:generalBounds}) is simply $E$ (the edge set of the input graph $G$) and we assume that we are given a partition of $E$ with $\ng$ disjoint groups, i.e., $\Gamma = (E_1,\ldots,E_{\ng})$, where $E = E_1 \sqcup \cdots \sqcup E_{\ng}$. With this notation in mind, all the optimization problems defined in Section~\ref{sec:notation} (e.g., $\SFMP(G,\Gamma),...$) are defined with respect to the Max-Cut problem with edge utilities.

% we will define $\SFMP(G, \Gamma)$, $\DFMP(G, \Gamma)$, $\MP(G)$, $\SFMV(G, \Gamma)$, $\DFMV(G, \Gamma)$, and $\MV(G)$ as the optimization problems introduced in Section~\ref{sec:generalBounds} with respect to the Max-Cut problem with edge utilities.\textcolor{red}{ew}

% $$\gs = \gs_1 \sqcup \gs_2 \sqcup \cdots \sqcup \gs_{\ng}$.

Now, each individual $e \in \gs$ obtains utility $X_e$ (Equation \ref{eqn:edgeCrossCutIndicator}) depending on whether or not $e$ crosses the cut. Thus, using the notation from Section \ref{sec:generalBounds},  for every $x \in \mathcal{X}$ and every $e \in \gs$, we have
$f_x(\{e\}) = X_e$. As an abuse of notation, we will often write $f_x(e)$ for $f_x(\{e\})$. For every $x \in \mathcal{X}$, the general set function is defined by simply summing the individual utilities, i.e., for any $F \subseteq \gs$,
\begin{equation} \label{eqn:fx-edge-def}
f_x(F) = \sum_{e \in F} f_x(e) = \sum_{e \in F} X_e.    
\end{equation}

Defined this way, it becomes clear that the corresponding MV (maximum value) is exactly the Max-Cut objective, i.e., Max-Cut maximizes the social welfare of the group through the lens of utilitarianism.

\begin{observation} \label{obs:MV-is-MC}
    For fair Max-Cut with edge utilities, we have that $\MV = \MC(G)$.
\end{observation}

\subsubsection{Comparing Optimal Static and Random Solutions}

Recall that in Section~\ref{sec:generalBounds}, we defined several objectives for combinatorial optimization problems. As noted in Observation~\ref{obs:MV-is-MC} above, if we let $f_x(F)$ be the number of edges in $F$ that cross the cut $x$, then the Maximum Value (MV) problem reduces to Max-Cut. 

The maximin-fair version of MV, however, is arguably not a useful notion of fairness. For example, consider the complete bipartite graph $K_{n,n} = (V,E)$ with groups $\gs_1 = \{e\}$ and $\gs_2 = E \setminus\{e\}$, for some edge $e \in E$. Then $f_x(\gs_1) \leq 1$, whereas the bipartition $y$ achieves $f_y(\gs_2) = n^2-1$. Thus, it is impossible to approach equality with respect to this {\it unnormalized} notion of fairness. For this reason, we will restrict our attention to the Maximum Proportion (MP) and its fair variants in this section.

% {\bf Notes}
% \begin{itemize}
%     \item Examples showing the gap between DF-MP and MP? Upper bound the absolute value of the gap? 
%     \item Triangle with weights 1, 1, 2 is an example of a gap between DF-MP and MP. Can we think of an unweighted example?
% \end{itemize}

When the partition of the edges $E$ only has one part (i.e., $\ng=1$ and $\gs_1 = E$), then it is clear that $\SFMP = \DFMP = \MP$. These quantities are also identical when $G$ is bipartite as seen in the proposition below.

\begin{proposition}
    \label{thm:propsEqualWhenBipartite}
    If $G = (V,E)$ is bipartite, then $\SFMP = \DFMP = \MP = 1.$
\end{proposition}

\begin{proof}
    Let $S\subseteq V$ be a bipartition of $G$. Then $f_S(E_i) = |E_i|/|E_i| = 1$ for all $i$, implying that $\SFMP(G,\Gamma) = \DFMP(G,\Gamma)=1$. Since $f_S(E) = |E|/|E|=1$ as well, this proves the claim.
    % Clearly, $\SFMP = \DFMP = \MP$ are all bounded above by 1. We also know that $\SFMP \leq \DFMP \leq \MP$; thus it suffices to show that $\SFMP \geq 1$ to complete the proof. Since $G$ is bipartite, then by definition, there must be a cut $S$ %with $\text{cut}(S) = |E|$; i.e., $S$ cuts every edge in the graph. 
    % which every edge spans. Since $S$ cuts every edge, it must be that $\text{SFMP} \geq \min_i |\gs_i|/|\gs_i| = 1$.
\end{proof}

\begin{proposition}
    Let $G$ be a graph with a single-edge partitioning. Then $\SFMP$ can only take on one of two values: 1 in the case that $G$ is bipartite, and 0 otherwise.
    % $$\SFMP = \begin{cases} 0, & G \text{ is not bipartite} \\ 1, & G \text{ is bipartite}\end{cases}.$$
\end{proposition}
\begin{proof}
    Let $\Gamma=\{E_1,\ldots,E_{|E|}\}$ be the edge partition, where $E_i = \{e_i\}$ for all $i$. If $G$ is bipartite, then we know $\SFMP = 1$ (regardless of the how the edges are partitioned) from Proposition~\ref{thm:propsEqualWhenBipartite}. Now suppose $G$ is not bipartite and let $S$ be any cut. Then, there exists an edge $e_j \in E$  that is not cut by $S$, implying that $\min_i f_S(E_i)/|E_i| = 0$. Since $S$ was an arbitrary cut, it follows that $\SFMP = \max_S \min_i f_S(E_i)/|E_i| = 0$.
    % Let $E_j = \{e\}$ be the group containing this edge. Then $\SFMP \geq \min_i f_S$
    % Since we have a single-edge partitioning, then there exists a part of the partition $E_j = \{e\}$ for some $j \in [m]$. Clearly, since $E_j$ only contains the edge $e$ which is not cut by $S$, then $\rho_j(S) = 0$ and hence $\text{SF-}\rho(S) = \min_i \rho_i(S) \leq \rho_j(S) = 0$. Clearly, $\text{SF-}\rho(S)$ is non-negative, and hence it must be that $\text{SF-}\rho(S) = 0$  Observe that this holds for any choice of $S$, hence $\SFMP(G) = \max_{S \subseteq V} \text{SF-}\rho(S) = \max_{S \subseteq V} 0 = 0.$
\end{proof}

In the remainder of this subsection, we will discuss the implications of the general bounds of Section~\ref{sec:generalBounds} on Max-Cut with edge utilities, and how the general bounds can be improved. First, observe that $f_x$, as defined in (\ref{eqn:fx-edge-def}), is nonnegative and superadditive. Thus, we have the following:
\begin{corollary}
For the Max-Cut problem with graph $G$ and edge partition $\Gamma = \{E_1,\ldots,E_{\ng}\}$, we have that $\SFMP(G,\Gamma) \leq \DFMP(G,\Gamma) \leq \MP(G).$

\end{corollary}

Moreover, as seen in Table \ref{tab:gaps}, in the case of Max-Cut with edge utilities, these bounds can be tightened as seen in Propositions \ref{prop:mp-dfmp-edge} and \ref{prop:dfmp-sfmp-edge} below. The proofs for these propositions can be found in Appendix \ref{sec:proofs_edgeUtilities}.

\begin{proposition} \label{prop:mp-dfmp-edge}
For any graph $G$ and edge partition $\Gamma$, it must be that $\MP - \DFMP \leq 1/2$. Moreover, for any $\varepsilon > 0$, there exists a graph $G$ and edge partition $\Gamma$ for which $\MP - \DFMP > \frac{1}{2} - \varepsilon$.
\end{proposition}

We similarly provide a tight bound on the gap between the $\SFMP$ and the $\DFMP$ for Max-Cut with edge utilities.

\begin{proposition} \label{prop:dfmp-sfmp-edge}
For any $\varepsilon > 0$, there exists a graph $G$ and edge partition $\Gamma$ for which $\DFMP - \SFMP > 1-\varepsilon$.
\end{proposition}

Finally, we discuss the sub-problem bound (Theorem~\ref{thm:subproblem-bounds}) in the context of Max-Cut with edge utilities. 

\begin{corollary} \label{cor:subgraph-edge}
Let $G=(V,E)$ be a graph with edge partition $\Gamma = \{E_1,\ldots,E_{\ng}\}$, and let $H=(V,E')$ be a subgraph of $G$ whose edge set is $\bigcup_{j=1}^k E_{s_j}$, where $s_j \in [\ng]$ for all $j$. Then, $\DFMP(G,\Gamma) \leq \MP(H).$
\end{corollary}

\begin{proof}
This follows immediately from Theorem~\ref{thm:subproblem-bounds} by setting $f'_S(F) = f_S(F)$ for every cut $S$ and every subset of edges $F$.
\end{proof}

The above corollary implies that the $\DFMP$ of a graph is bounded by the $\MP$ of each of its edge groups. For example, this gives us the following bound:

\begin{corollary}
    Let $G$ be a graph with edge partition $\Gamma=\{E_1,\ldots,E_{\ng}\}$. If, for any $i$, $E_i$ is a 3-cycle or an edge-disjoint union of 3-cycles, then $\DFMP \leq 2/3$.
\end{corollary}

The bound in Corollary~\ref{cor:subgraph-edge} begs the question: does there always exist a subgraph $H$ (of the form described in Corollary~\ref{cor:subgraph-edge}) for which $\DFMP(G) = \MP(H)$? We demonstrate an example in Appendix \ref{sec:proofs_edgeUtilities} which answers this question in the negative.

In the next section, we describe algorithms for the fair Max-Cut problem.

\subsubsection{Algorithms}
\label{subsubsec:edge-algorithms}

Since the $\DFMP$ and $\SFMP$ problems reduce to Max-Cut in the case where $\Gamma = \{E\}$, these problems are NP-hard. That said, we discuss a few algorithms in order to showcase the challenges of algorithm design for this problem.

First, consider the naive random cut, wherein each vertex is included in the cut $S$ independently with probability 1/2. As shown in Appendix~\ref{sec:naiveRandomCut}, under this algorithm, $\E_S f_S(E_i)/|E_i| = 1/2$ for each edge group $E_i$. In some cases, this could be viewed as a fair outcome, since each group is treated equally in expectation. However, as shown in Appendix~\ref{sec:naiveRandomCut}, different groups face different variances in $f_S(E_i)/|E_i|$, which could be viewed as unfair. Moreover, from an overall utility standpoint, this algorithm is far from optimal.

Another algorithm of note is the Goemans-Williamson (GW) algorithm \cite{goemans1995improved} which yields an expected approximation ratio of $\alpha \approx 0.878$, i.e., for all graphs $G$ ,  $\E[\text{ALG}(G)] \geq \alpha \MC(G)$ where $\E[\text{ALG}(G)]$ is the expected number of edges returned by the algorithm on graph $G$. In summary, the GW algorithm solves an semidefinite programming (SDP) relaxation of the Max-Cut objective, yielding $n$ points (corresponding to vertices) on the surface of a sphere in $n$-dimensions, then, a random hyperplane through the sphere's center is used to partition the vertices into two groups. Amongst all polynomial-time approximation algorithms for Max-Cut, the GW algorithm has the best-known approximation ratio and furthermore, under the Unique Games Conjecture,  $\alpha \approx 0.878$ is the best constant-factor approximation ratio that one can hope to achieve, assuming $\mathbf{P} \neq \mathbf{NP}$ \cite{khot2007optimal}. Despite this, the distribution of cuts returned by the GW algorithm are catastrophic in terms of dynamic fairness, in some cases, yielding a dynamic fairness proportion value of 0 for some graph instances (with non-zero \DFMP) as seen in Example \ref{ex:GW_gives_bad_DF} in Appendix \ref{sec:proofs_edgeUtilities}.

Finally, we note that the $\DFMP$ problem with graph $G$ and partition $\Gamma = \{E_1,\ldots,E_{\ng}\}$ can be solved exactly using the following linear program:
\begin{align*}
    \text{max} \ & z \\
    \text{s.t.} \  & z \leq \mathbb{E} \bigg[ \frac{f_S(E_i)}{|E_i|} \bigg] & \mbox{for every } i \in [\ng] \\
    & \sum_{S \subseteq V} p_S = 1\\
    & p_S \geq 0 & \mbox{for every } S \subseteq V. 
\end{align*}
The variables of this LP are $p_S$ (for each $S \subseteq V$), which represents the probability of selecting the cut $S$, and an auxiliary variable $z$, which incorporates the inner minimization. While this approach is perfectly fair with respect to the $\DFMP$ objective, it is inefficient, since the LP has exponentially many variables. We present a concrete example of this linear program on a 4-node graph in Example \ref{ex:lp_edge_utilities} in Appendix \ref{sec:proofs_edgeUtilities}.

\subsection{Node Utilities}
\label{sec:nodeUtilities}

We now express fair Max-Cut with node utilities using the notation in Section \ref{sec:generalBounds}. As in the edge utility case, we still have $\mathcal{X} = \mathcal{P}(V)$ as our decision set. Our set of individuals $\gs$ (as described in Section \ref{sec:generalBounds}) is now $V$ (the vertex set of the input graph $G$) and we assume that we are given a partition of $V$ with $\ng$ disjoint groups, i.e., $V = V_1 \sqcup V_2 \sqcup \cdots \sqcup V_\ng$.

Now, each individual corresponding to $v \in V$ obtains utility proportional to the number of its incident edges that are crossing the cut. To aid in defining this utility, for each $v \in V$, we define $N(v)$ as the set of vertices that are neighbors of $v$, i.e., $N(v) = \{u \in V: (u,v) \in E\}$. Thus, using the notation from Section \ref{sec:generalBounds},  for every $x \in \mathcal{X}$ and every $v \in V$, we have
\begin{equation}
\label{eqn:nodeUtility}
f_x(\{v\}) = \frac{1}{\Delta(G)}\sum_{u \in N(v)} X_{uv},\end{equation}
where $\Delta(G)$ is the maximum degree of any vertex in $G$ and $X_{uv} := X_e$ where $e =(u,v)$ as defined in Section \ref{sec:edgeUtilities}. As an abuse of notation, we will often write $f_x(v)$ for $f_x(\{v\})$. For every $x \in \mathcal{X}$, the general set function is defined by simply summing the individual utilities, i.e., for any $U \subseteq V$,
\begin{equation} \label{eqn:node-utility-def}
f_x(U) = \sum_{v \in U} f_x(v) = \frac{1}{\Delta(G)} \sum_{v \in U} \sum_{u \in N(v)} X_{uv}.
\end{equation}
We define the optimization problems of Section~\ref{sec:notation} with respect to these utility functions. 

Defined this way, one can show that the MV (maximum value) is proportional to the Max-Cut objective, i.e., Max-Cut maximizes the social welfare of the group through the lens of utilitarianism. See Appendix~\ref{sec:nodeUtilities} for the proof.

\begin{proposition}
\label{thm:nodeUtilPropToMaxcut}
    For fair Max-Cut with node utilities, we have,$$\MV = \frac{2}{\Delta(G)}\cdot \MC(G).$$    
\end{proposition}

\begin{proof}[Proof of Proposition \ref{thm:nodeUtilPropToMaxcut}]
    Observe,
    \[
 \MV = \max_{x \in \mathcal{X}} f_x(V)
        = \max_{x \in \mathcal{X}} \frac{1}{\Delta(G)} \sum_{v \in V} \sum_{u \in N(v)} X_{uv} 
        \overset{(*)}{=} \max_{x \in \mathcal{X}} \frac{2}{\Delta(G)} \sum_{e \in E} X_e
        = \frac{2}{\Delta(G)} \MC(G),
    \]
    % \begin{align*}
    %     \MV &= \max_{x \in \mathcal{X}} f_x(V)\\
    %     &= \max_{x \in \mathcal{X}} \sum_{v \in V} f_x(v) \\
    %     &= \max_{x \in \mathcal{X}} \frac{1}{\Delta(G)} \sum_{v \in V} \sum_{u \in N(v)} X_{uv} \\
    %     &= \max_{x \in \mathcal{X}} \frac{2}{\Delta(G)} \sum_{e \in E} X_e \tag{$*$}\\ 
    %     &=  \frac{2}{\Delta(G)} \max_{x \in \mathcal{X}} \sum_{e \in E} X_e \\ 
    %     &= \frac{2}{\Delta(G)} \MC(G).
    % \end{align*}
where the equation with $(*)$ follows by considering the well-known handshaking lemma applied to the subgraph consisting of edges that cross the cut.
\end{proof}

Additionally, because of the factor of $\frac{1}{\Delta(G)}$ in the node utility, the maximum proporition $\MP$ is a ``true'' proportion in the sense that it lies between 0 and 1.

\begin{proposition}
\label{thm:utilBetweenZeroAndOne}
    For fair Max-Cut with node utilities, we have $0 \leq \MP \leq 1$.
\end{proposition}
\begin{proof}
    It is clear that $\MP \geq 0$. Observe that for any $v \in V$, $\sum_{u \in N(v)} X_{uv} \leq \Delta(G)$ and hence $f_x(v) = \frac{1}{\Delta(G)}\sum_{u \in N(v)} X_{uv} \leq 1$. Thus,
    $$\MP = \max_{x \in \mathcal{X}} \frac{f_x(V)}{|V|} = \max_{x \in \mathcal{X}} \frac{\sum_{v \in V} f_x(v)}{|V|} \leq \max_{x \in \mathcal{X}} \frac{\sum_{v \in V} 1}{|V|} = 1.\eqno\qedhere$$
\end{proof}
For any graph, the node utility defined in Equation \ref{eqn:nodeUtility} may not be the only notion of utility that one may want to consider. For example, one may want to define alternative node utilities $\{\tilde{f}_x\}_{x \in \mathcal{X}}$ in a way that is relative to the \emph{degree} of each vertex, i.e., 
$$\tilde{f}_x(\{v\}) = \frac{1}{\text{deg}(v)}\sum_{u \in N(v)} X_{uv},$$
and for all $A \subseteq V$, $\tilde{f}_x(A) = \sum_{v \in A} \tilde{f}_x(v).$

If we choose $\{\tilde{f}_x\}_{x \in \mathcal{X}}$ instead of $\{f_x\}_{x \in \mathcal{X}}$ for our set functions, then the Maximum Proportion, $\MP$, still lies between 0 and 1 as in Proposition \ref{thm:utilBetweenZeroAndOne}; however, the Maximum Value, $\MV$, is no longer necessarily\footnote{It should be noted that in the special case that $G$ is a regular graph and the set functions  $\{\tilde{f}_x\}_{x \in \mathcal{X}}$ are used, then it is once again the case that $\MV \propto \MC$  as $\tilde{f}_x = f_x$ for every $x \in \mathcal{X}$.} proportional to the maximum cut (as in Proposition \ref{thm:nodeUtilPropToMaxcut}), meaning that maximin fairness with set functions $\{\tilde{f}_x\}_{x \in \mathcal{X}}$ could not as easily be interpreted as some notion of ``fair Max-Cut.'' Nonetheless, maximin fairness with set functions $\{\tilde{f}_x\}_{x \in \mathcal{X}}$ may still be of interest in its own right.%; however, this is outside the scope of this work (which is geared towards notions of fair Max-Cut).

\subsubsection{Comparing Optimal Static and Random Solutions}

As we did in Section~\ref{sec:edgeUtilities} for the case of edge utilities, we now compare static and dynamic fair solutions to Max-Cut with node utilities. As we will see, the bounds obtainable for the case of node utilities tend to depend more heavily on the structure of the graph than those for the case of edge utilities.

To begin, we show that for regular bipartite graphs, all three optimization problems equal 1. Note that the analogous result for edge utilities (Theorem~\ref{thm:propsEqualWhenBipartite}) did not require regularity.

% similar to Theorem \ref{thm:propsEqualWhenBipartite},when $G$ is a regular graph, we have that $G$ being bipartite implies that the various notions of proportion are all equal to 1.
\begin{proposition}
    \label{thm:vertex_propsEqualWhenBipartite}
    If $G = (V,E)$ is bipartite and regular, and $\Gamma = \{V_1,\ldots,V_\ng\}$ is a vertex partition, then $\SFMP(G,\Gamma) = \DFMP(G,\Gamma) = \MP(G) = 1$ when vertex utilities are used.
\end{proposition}

\begin{proof}
Let $S$ be a bipartition. Then for any vertex group $V_i$, we have that
\[
f_S(V_i) = \frac{1}{|V_i|} \sum_{v \in V_i} \frac{\deg v}{\Delta(G)} \overset{(a)}{=} \frac{1}{|V_i|} \sum_{v \in V_i} 1 = 1,
\]
where (a) follows from the regularity assumption. It follows that $\SFMP = \DFMP = 1$. Moreover, we have that $\MP(G) \geq f_S(V) = \frac{1}{|V|}\sum_{v \in V} \deg (v)/\Delta(G) = 1$.
\end{proof}

Recall that for non-bipartite graphs with a single-edge partition, the static fair maximum proportion with respect to edge utilities is 0. This is not true in the case of node utilities; instead, we get a weaker bound. In particular, note that for a non-bipartite graph $G$ and any cut $S$, there must be some node for which not every incident edge is cut. We use this observation to obtain the following bound:

\begin{proposition}
Let $G$ be a non-bipartite graph with a single-vertex partitioning. Then $\SFMP \leq \frac{\Delta(G)-1}{\Delta(G)}$.
\end{proposition}
\begin{proof}
    Let $(S, V \setminus S)$ be any cut. Since $G$ is non-bipartite, then there must be two adjacent vertices on the same side of the cut, say $u,v \in S$. Note that for $v$, its node utility is
    \begin{align*}f_S(v) &= \frac{1}{\Delta(G)} \sum_{w \in N(v)} X_{vw} = \frac{1}{\Delta(G)}\bigg(\underbrace{X_{uv}}_{=0} + \sum_{w \in N(v)\setminus \{u\}} \underbrace{X_{vw}}_{\leq 1}\bigg) 
    %&\leq  \frac{1}{\Delta(G)}\bigg(\sum_{w \in N(v)\setminus \{u\}}1\bigg) 
    \leq  \frac{\Delta(G)-1}{\Delta(G)}.
    \end{align*}
    Recall that $\Gamma$ is a single-vertex partitioning, thus, $$\min_{i \in [\ng]} \frac{f_S(V_i)}{|V_i|} = \min_{w \in V} \frac{f_S(w)}{1} \leq  f_S(v)  \leq \frac{\Delta(G)-1}{\Delta(G)}.$$ Since $S$ was an arbitrary cut, then $\SFMP =  \max_{S \in \mathcal{X}}\min_{i \in [\ng]} \frac{f_S(V_i)}{|V_i|} \leq \frac{\Delta(G)-1}{\Delta(G)}.$
\end{proof}

As we did in Section~\ref{sec:edgeUtilities}, we will now discuss the implications of the general bounds of Section~\ref{sec:generalBounds} on Max-Cut with node utilities. Since the node utility functions $f_x$ defined in (\ref{eqn:node-utility-def}) are nonnegative and superadditive, we have the following:

\begin{corollary}
For the Max-Cut problem with graph $G$ and node partition $\Gamma = \{V_1,\ldots,V_{\ng}\}$, we have that $\SFMP(G,\Gamma) \leq \DFMP(G,\Gamma) \leq \MP(G).$
\end{corollary}

In contrast to the case of edge utilities, the gap between the maximum proportion and the dynamic fair maximum proportion can be arbitrarily large in the case of node utilities.

\begin{proposition} \label{prop:mp-dfmp-node}
For any $\varepsilon > 0$, there exists a graph $G$ and node partition $\Gamma$ for which $\MP(G) - \DFMP(G,\Gamma) > 1-\varepsilon$.
\end{proposition}

The proof of this proposition follows from the observation that $\DFMP(G) \leq \min_i \frac{\Delta(V_i)}{\Delta(G)}$ and can be found in Appendix~\ref{sec:proofs_nodeUtilities}. On the other hand, in contrast to the bound proved for edge utilities, the gap between $\DFMP(G,\Gamma)$ and $\SFMP(G,\Gamma)$ cannot be arbitrarily large in the case of node utilities. In particular, the gap is bounded by 1/2:

\begin{proposition} \label{prop:dfmp-minus-sfmp-node}
For any $\varepsilon > 0$, there exists a graph $G$ and node partition $\Gamma$ for which $\DFMP(G,\Gamma) - \SFMP(G,\Gamma) > \frac{1}{2}-\varepsilon$. Moreover, for any graph $G$ and node partition $\Gamma$, we have that $\DFMP(G,\Gamma) - \SFMP(G,\Gamma) \leq \frac{1}{2} \big( \min_i \frac{1}{|V_i|} \sum_{u \in V_i} \frac{\deg(u)}{\Delta(G)} \big) \leq \frac{1}{2}$.
\end{proposition}

A proof sketch is presented below, and the full proof is in Appendix~\ref{sec:proofs_nodeUtilities}.

\begin{proof}[Proof sketch]
We begin with the first statement. Consider the $(2n+1)$-cycle $C_{2n+1}$ (as shown below for $n=2$) with $2n+1$ singleton vertex groups.

% a node partition $\Gamma_{2n+1}$ containing $2n+1$ singleton vertex sets $V_i = \{i-1\}$ for $i=1,\ldots,2n+1$. This graph is depicted below for $n=2$.

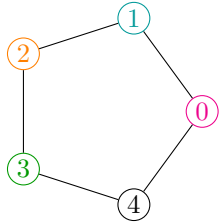
\begin{center}
\begin{tikzpicture}

\tikzstyle{vertex}=[draw, circle, color = black, text opacity = 1, inner sep = 1.4pt]

\node[vertex, magenta] (v0) at ({360/5 * 0}:1.3cm){0};
\node[vertex, dcyan] (v1) at ({360/5 * 1}:1.3cm){1};
\node[vertex, orange] (v2) at ({360/5 * 2}:1.3cm){2};
\node[vertex, dgreen] (v3) at ({360/5 * 3}:1.3cm){3};
\node[vertex] (v4) at ({360/5 * 4}:1.3cm){4};

\draw (v0) -- (v1);
\draw (v1) -- (v2);
\draw (v2) -- (v3);
\draw (v3) -- (v4);
\draw (v4) -- (v0);

\end{tikzpicture}
\end{center}

Since the cycle is odd, the static fair maximum proportion is $1/2$. To bound the $\DFMP$, note that for each edge $e$, there is a maximum cut $S_e$ which does not cut $e$. By choosing one of these cuts uniformly at random, we can bound $\min_i \E[f_{S_e}(V_i)]/|V_i|$ by $1-o(1)$, implying that $\DFMP-\SFMP \to 1/2$ as $n \to \infty$. 

Tightness follows from the two bounds below, which we prove in the appendix:
\begin{align}
\DFMP(G,\Gamma) &\leq \min_i \frac{1}{|V_i|} \sum_{u \in V_i} \frac{\deg(u)}{\Delta(G)} \label{eqn:dfmp-upper-bound-body} \\
\SFMP(G,\Gamma ) &\geq \min_i \frac{1}{|V_i|} \sum_{u \in V_i} \frac{\deg(u)}{2\Delta(G)} \label{eqn:sfmp-lower-bound-body}
\end{align}
The top bound (\ref{eqn:dfmp-upper-bound-body}) simply comes from the observation that number number of edges incident to $u$ that span a cut $S$ is at most $\deg u$. The second bound (\ref{eqn:sfmp-lower-bound-body}) comes from the observation that the local search algorithm for Max-Cut results in a cut $S$ for which $f_S(V_i) \geq \sum_{u \in V_i} \frac{\deg u}{2\Delta(G)}$. 
\end{proof}

Finally, we translate the subproblem bound on Section~\ref{sec:generalBounds} to a subgraph bound for Max-Cut with node utilities. Unline the subgraph bound for edge utilities, the bound for node utilities has an additive error term due to the edges spanning multiple vertex groups.

\begin{corollary} \label{cor:subgraph-node}
Let $G=(V,E)$ be a graph with node partition $\Gamma = \{V_1,\ldots,V_{\ng}\}$, and $H$ a vertex-induced subgraph of $G$ whose vertex set is $\bigcup_{j=1}^{\ng'} V_{s_j}$, where $s_j \in [\ng]$ for all $j$. Then, 
\[
\DFMP(G,\Gamma) \leq \MP(H) + \frac{1}{|V(H)|}\sum_{j=1}^{\ng'} e(V_{s_j}, \overline{V(H)})~,
\]
where $e(V_{s_j}, \overline{V(H)})$ is the number of edges with one endpoint in $V_{s_j}$ and another in $V \setminus V(H)$.
\end{corollary}

\begin{proof} 
Let $f'_S(U) = \sum_{u \in U} \sum_{w \in N_H(u)} X_{uw}$ for every cut $S$ and every subset of nodes $U$, where $N_H(u)$ is the set of neighbors of $u$ in $H$. Then 
\[
f'_S(V_{s_j}) + e(V_{s_g}, \overline{V(H)}) \geq f_S(V_{s_j})
\]
for every $S \subseteq V$ and every $j \in [\ng']$. The claim thus follows from Theorem~\ref{thm:subproblem-bounds}.
\end{proof}

\subsubsection{Algorithms}

In this section, we discuss two approaches to the fair Max-Cut problem with node utilities. Of course, any Max-Cut algorithm can be used 

We can do an analysis similar to that in Section \ref{sec:naiveRandomCut} to see how this vertex notion of Max-Cut fairness behaves with respect to the naive random algorithm. Let $S$ be the (random) cut obtained by the naive random algorithm, then we have the following for any group $V_i$:
\begin{align*}
    \mathbb{E} \left[\frac{f_S(V_i)}{|V_i|} \right] %&= \mathbb{E}\left[\frac{1}{|V_i|} \sum_{v \in V_i} \cut_v(S)/\Delta(G)\right] \\
    &= \mathbb{E}\bigg[\frac{1}{|V_i|\Delta(G)} \sum_{v \in V_i} \sum_{u \in N(v)} X_{uv}\bigg] %=\frac{1}{|V_i|\Delta(G)} \sum_{v \in V_i} \sum_{u \in N(v)} \mathbb{E}\left[X_{uv}\right] \\
    =\frac{1}{2|V_i|\Delta(G)} \sum_{v \in V_i} \text{deg}(v). %\tag{as $\mathbb{E}\left[X_{uv}\right] = 1/2$}
\end{align*}

If $G$ is $k$-regular, then $\mathbb{E}[f_S(V_i)/|V_i|] = \frac{1}{2|V_i|\Delta(G)} \sum_{v \in V_i} \text{deg}(v) = 1/2$. Otherwise,
\begin{align*}
    \frac{1}{2\Delta(G)} \bigg( \frac{1}{|V_i|}\sum_{v \in V_i} \text{deg}_{G[V_i]}(v)\bigg) \leq \mathbb{E} \left[\frac{f_S(V_i)}{|V_i|} \right] %&= \mathbb{E}\left[\frac{1}{|V_i|} \sum_{v \in V_i} \cut_v(S)/\Delta(G)\right] \\
    %&=\frac{1}{2|V_i|\Delta(G)} \sum_{v \in V_i} \text{deg}(v) 
    \leq \frac{1}{2\Delta(G)} \bigg( \frac{1}{|V_i|}\sum_{v \in V_i} \text{deg} (v)\bigg),
    %&= \frac{1}{2\Delta(G)}\Big(\text{average degree in the induced subgraph $G[V_i]$}\Big)
\end{align*}
where $\deg_{G[V_i]}(v)$ is the degree of $v$ in the induced subgraph $G[V_i]$. Thus, for an arbitrary graph $G$, the naive random cut yields an expected proportion for group $V_i$ which is sandwiched between $1/2\Delta(G)$ times the average degree in $G[V_i]$ and $1/2\Delta(G)$ times the average degree of $V_i$ in $G$. Here, we can see that an artifact of using node utilities is that groups with low degree (compared to $\Delta(G)$) will incur lower utility.

As with the case of edge utility, one can find $\DFMP(G,\Gamma)$ exactly using a linear program. In particular, letting $p : 2^V \to [0,1]$ denote the distribution over cuts, $\DFMP(G,\Gamma)$ is
\begin{align*}
    \text{max} \ & z \\
    \text{s.t.} \  & z \leq \sum_{S \subseteq V}\bigg( \frac{\sum_{u \in V_i} \sum_{w \in N(u)} X_{uw}}{|V_i|\Delta(G)} \bigg) p_S  & \forall i \in [\ng] \\
    & \sum_{S \subseteq V} p_S = 1\\
    & p_S \geq 0 & \forall S \subseteq V. 
\end{align*}
While this method is inefficient, it will produce the optimal distribution over cuts.

% \begin{itemize}
    % \item Problem formulation--Reuben
    % \begin{itemize}
    %     \item Note: while we could define $f_S(V_i) = \sum_{v \in V_i} \frac{\mbox{cut}_v(S)}{\deg v}$ (and such a function would still meet the condition of Theorem~\ref{thm:general-bound}), we would lose the connection to Max-Cut. That is, $f_S(V)$ would not necessarily be equal to $\operatorname{MP}(G)$.
    % \end{itemize}
    % \item Bounds: --Jad
    % \begin{itemize}
    %     \item basic observations (bipartite, naive random cut, \textcolor{red}{node partition--Reuben})
    %     \item tightness of SFMP $\leq$ DFMP $\leq$ MP bounds
    %     \item subgraph bound
    % \end{itemize}
    % \item LP Formulation--Jad \textcolor{red}{Adding a colored note here to make sure this isn't forgotten}
% \end{itemize}

\section{Discussion}
\label{sec:discussion}
In this paper, we introduced a general class of combinatorial optimization problems which we framed as offline problems. That is, we presented the problems of Section~\ref{sec:notation} as one-shot optimization problems, in which a single decision is made, and a single utility is incurred. However, in some settings, decisions must be made iteratively over time. In such settings, employing the static-fair optimal solution at each iteration will result in a maximin objective of $\SFMP$ (or $\SFMV$). Another approach, however, would be to sample from a distribution independently at each iteration. Doing so with the dynamic-fair optimal solution will result in an expected maximin objective of $\DFMP$ (or $\DFMV$), which is an improvement over the static fair solution. Moreover, varying decisions over time (such as in the latter approach) can be used to ensure that one group isn't routinely disadvantaged. Thus, finding good distributions over solutions is closely related to the concept of {\it fairness over time}, which has gained traction in the fairness community. These temporal notions of fairness constrain sequences of decisions (e.g., requiring that different groups incur similar sliding-window-average utilities). We leave the question of combinatorial optimization under temporal maximin constraints as an open question. 

The decision set for our dynamic fair optimization problems is the set of distributions over $\mathcal{X}$, and the decision set for our static fair optimization problems is the set of vertices of this simplex. Thus, the former decision set can be viewed as a set of {\it fractional solutions}, and the latter can be viewed as a set of {\it integral solutions}. The tight bounds on the gaps between the dynamic fair and static fair optima can thus be interpreted as {\it integrality gaps} for their respective optimization problems. In general, such integrality gaps have not been quantified in maximin-fair combinatorial optimization problems. We leave the characterization of these integrality gaps as an open question.

% Combinatorial optimization problems can be expressed as mathematical optimization problems with integral variables. Since such problems are often difficult to solve, they are often relaxed by removing the integrality constraint. The value of the solution to the relaxed problem (i.e., the fractional solution) is a bound on the value of the solution to the original problem. Studying the gap between these two values is central to combinatorial optimization. The integrality gap, however, has not been quantified in maximin-fair combinatorial optimization problems in general. In this paper, we proved tight bounds on the integrality gaps of two fair variants of the Max-Cut problem. It is an open question to characterize the integrality gaps of other maximin combinatorial optimization problems.

The algorithms we have considered in this work are either known algorithms (used for the non-fair versions of the problems), naive meta-algorithms constructed from such already-known algorithms, or exact solvers via linear programming that require an exponential number of variables. Clearly, for \SFMP{} or \DFMP{} problems, more sophisticated algorithms are needed if one hopes to have better performance guarantees in practice. % theoretical or practical static or dynamic fairness values; w
We believe that this is an interesting direction for future work. In particular, \emph{iterative} approaches---which may be combinatorial in nature or based on tuning parameters of an algorithm---may be fruitful.%, i.e., one iteratively updates (via solution modifications or tuning of the algorithm's parameters) in a way that encourages the algorithm to try to make the worst-off group better off than they were in the previous iteration.

Next, we turn to the field of quantum computing. Determining for which problems we can expect current or future quantum computers to have some sort of advantage over classical computers (either theoretically or in practice) is still an open and active area of research. From a theoretical perspective, assuming the Unique Games Conjecture and certain widely-believed complexity theoretic assumptions, problems such as the (vanilla) Max-Cut problem are not a good candidate since no classical or quantum algorithm can beat the classical Goemans-Williamson algorithm in terms of approximation ratio \cite{khot2007optimal}. 
%Recently, some have turned towards thinking of specific settings for problems for which quantum is likely to have an advantage, this is the motivation for Kallaugher et al. to consider the Maximum-Directed-Cut problem in a \emph{streaming} model; with this streaming model, they demonstrated a quantum algorithm which achieves a theoretical exponential space advantage over classical algorithms \cite{kallaugher2024exponential}.
We believe that the \DFMP{} problem associated with Max-Cut, i.e., ``dynamic fair Max-Cut,'' is a promising candidate since we are shifting our setting from single solutions to \emph{distributions} of solutions. It is well-known that quantum circuits have the ability to produce distributions of bitstrings that are difficult for classical algorithms to produce; in particular, this result holds for the Quantum Approximate Optimization Algorithm \cite{farhi2014quantum,farhi2016quantum} which is a general quantum algorithm used for solving problems in combinatorial optimization.

\section{Acknowledgments}
\label{section:acknowledgments}
This work was supported by the U.S. Department of Energy through the Los Alamos National Laboratory. Los Alamos National Laboratory is operated by Triad National Security, LLC, for the National Nuclear Security Administration of U.S. Department of Energy (Contract No. 89233218CNA000001). The research presented in this article was supported by the Laboratory Directed Research and Development program of Los Alamos National Laboratory under project number 20230049DR as well as by the NNSA's Advanced Simulation and Computing Beyond Moore's Law Program at Los Alamos National Laboratory. Report Number: LA-UR-24-30325.  

%%
%% Bibliography
%%

%% Please use bibtex, 

\bibliography{lipics-v2021-sample-article}

\appendix

\section{The naive random cut}
\label{sec:naiveRandomCut}
Let $\Gamma =\{E_1,\ldots,E_\gamma\}$ be an edge partition, where $E_i = \{e^i_j : j=1,\ldots,|E_i|\}$ and $X_j^i$ is the indicator of the $j$th edge in $E_i$ crossing the cut. Consider the algorithm that places each vertex uniformly at random on either side of the cut. Let $S$ be this cut. In this case,
\[
\mathbb{E} \frac{f_S(E_i)}{|E_i|} = \frac{1}{|E_i|} \sum_{e \in E_i}^{|E_i|} \mathbb{E}X^i_j = \frac{1}{2}.
\]
To find the variance of this proportion, we must first find $\mathbb{E} [X^i_{j_1} X^i_{j_2}]$ for $1 \leq j_1 \ne j_2 \leq |E_i|$. If $e^i_{j_1}$ and $e^i_{j_2}$ are disjoint edges, then clearly $X^i_{j_1}$ and $X^i_{j_2}$ are independent, and $\mathbb{E} [X^i_{j_1} X^i_{j_2}] = \mathbb{E} [X^i_{j_1}] \mathbb{E}[X^i_{j_2}] = 1/4$. If they are not disjoint, then they share a vertex, as shown below.

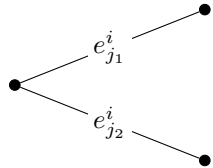
\begin{center}
\begin{tikzpicture}

\tikzstyle{vertex}=[draw, circle, color = black, fill = black, text opacity = 1, inner sep = 1.4pt]

\node[vertex] (v1) at (0,0) {};
\node[vertex] (v2) at (2.5,1) {};
\node[vertex] (v3) at (2.5,-1) {};

\draw (v1) -- node[inner sep=2pt, fill=white] {$e^i_{j_1}$} (v2);
\draw (v1) -- node[inner sep=2pt, fill=white] {$e^i_{j_2}$} (v3);

\end{tikzpicture}
\end{center}
In this case, both edges span the cut if and only if the right two vertices are on the opposite side of the left vertex, which happens with probability 1/4. Thus, in this case, $\mathbb{E}[X^i_{j_1} X^i_{j_2}] = 1/4$ as well. So, regardless of which edges $e^i_{j_1}$, $e^i_{j_2}$ are considered, we have that
\[
\mathbb{E}[X^i_{j_1} X^i_{j_2}] = \mathbb{E} [X^i_{j_1}] \mathbb{E}[X^i_{j_2}] = \frac{1}{4}.
\]
Hence,
\begin{align*}
    \operatorname{Var}\bigg(\frac{f_S(E_i)}{|E_i|}\bigg) &= \frac{1}{|E_i|^2} \bigg[ \sum_{j=1}^{|E_i|} \operatorname{Var}(X^i_j) \bigg] \\
    &= \frac{1}{|E_i|^2} \bigg[ \sum_{j=1}^{|E_i|} \mathbb{E}\Big[\Big(X^i_j-\frac{1}{2}\Big)^2\Big] \bigg] \\
    &= \frac{1}{|E_i|^2} \bigg[ \sum_{j=1}^{|E_i|} \Big(\mathbb{E}[(X^i_j)^2] - \mathbb{E}[X^i_j] + \frac{1}{4} \Big) \bigg] \\
    &= \frac{1}{4|E_i|}.
\end{align*}
So, while the mean value of $\frac{f_S(E_i)}{|E_i|}$ is consistent across groups, the variance decays with the group size. Smaller groups are therefore more likely to be subject to low $\frac{f_S(E_i)}{|E_i|}$ values.

\section{Proof of Theorem \ref{thm:subproblem-bounds}  }
\label{sec:proofs_section_3}
\begin{proof}[Proof of Theorem \ref{thm:subproblem-bounds}]
    First, let $D \in \argmax_{D \in \mathcal{D}} \min_{i \in [\ng]} \E_{x \sim D} f_x(\gs_i)$. Then,
\begin{align*}
    \MV' + \sum_{i=1}^{\ng'} \delta(\gs'_i) &= \max_{x \in \mathcal{X}} f_x'(\gs') + \sum_{i=1}^{\ng'} \delta(\gs'_i) \\
    &\geq \E_{x \sim D} f_x'(\gs') + \sum_{i=1}^{\ng'} \delta(\gs'_i) \\
    &\geq \E_{x \sim D} \sum_{i=1}^{\ng'} f_x'(\gs'_i) + \sum_{i=1}^{\ng'} \delta(\gs'_i) &\mbox{by superadditivity} \\
    &= \E_{x \sim D} \sum_{i=1}^{\ng'} f_x'(\gs'_i) + \delta(\gs'_i)   \\
    &\geq \E_{x \sim D} \sum_{i=1}^{\ng'} f_x(\gs'_i)  \\
    &= \sum_{i=1}^{\ng'} \E_{x \sim D} f_x(\gs'_i)   \\
    &\geq \min_{i \in [\ng']} \E_{x \sim D} f_x(\gs_i') &\mbox{by nonnegativity} \\
    &\geq \min_{i \in [\ng]} \E_{x \sim D} f_x(\gs_i) \\
    &= \DFMV &\mbox{since $\displaystyle D \in \argmax_{D \in \mathcal{D}} \min_{i \in [\ng]} \E_{x \sim D} f_x(\gs_i)$.}
\end{align*}

    Now, let $D \in \argmax_{D \in \mathcal{D}} \min_{i \in [\ng]} \E_{x \sim D} f_x(\gs_i)/|\gs_i|$. Then,
    \begin{align*}
    \MP' + \frac{\sum_{i=1}^{\ng'} \delta(\gs'_i)}{\sum_{i=1}^{\ng'} |\gs'_i|} &= \max_{x \in \mathcal{X}} \frac{f_x'(\gs')}{|\gs'|} + \frac{\sum_{i=1}^{\ng'} \delta(\gs'_i)}{\sum_{i=1}^{\ng'} |\gs'_i|} \\
    &\geq \E_{x \sim D} \frac{f_x'(\gs')}{|\gs'|} + \frac{\sum_{i=1}^{\ng'} \delta(\gs'_i)}{\sum_{i=1}^{\ng'} |\gs'_i|} \\
    &\geq \E_{x \sim D} \frac{\sum_{i=1}^{\ng'} f_x'(\gs'_i)}{\sum_{i=1}^{\ng'} |\gs'_i|} + \frac{\sum_{i=1}^{\ng'} \delta(\gs'_i)}{\sum_{i=1}^{\ng'} |\gs'_i|} &\mbox{by superadditivity} \\
    &= \E_{x \sim D} \frac{\sum_{i=1}^{\ng'} f_x'(\gs'_i) + \delta(\gs'_i)}{\sum_{i=1}^{\ng'} |\gs'_i|}   \\
    &\geq \E_{x \sim D} \frac{\sum_{i=1}^{\ng'} f_x(\gs'_i)}{\sum_{i=1}^{\ng'} |\gs'_i|}  \\
    &= \frac{\sum_{i=1}^{\ng'} \E_{x \sim D} f_x(\gs'_i)}{\sum_{i=1}^{\ng'} |\gs'_i|}   \\
    &\geq \min_{i \in [\ng']} \frac{\E_{x \sim D} f_x(\gs_i')}{|\gs'_i|} &\mbox{by nonnegativity} \\
    &\geq \min_{i \in [\ng]} \frac{\E_{x \sim D} f_x(\gs_i)}{|\gs_i|} \\
    &= \DFMP &\mbox{since $\displaystyle D \in \argmax_{D \in \mathcal{D}} \min_{i \in [\ng]} \E_{x \sim D} \frac{f_x(\gs_i)}{|\gs_i|}$.}
\end{align*}

\end{proof}
\section{Examples and Proofs from Section~\ref{sec:edgeUtilities}}
\label{sec:proofs_edgeUtilities}
\begin{example}[Example showing that $\DFMP(G) <  \min_H \text{MP}(H)$ is possible]
\label{ex:subgraphBound}
 Here, we show an example of a graph for which $\DFMP(G) <  \min_H \text{MP}(H)$, where the minimum is taken over all $H$ that satisfy the conditions of the above theorem. This example uses the minimum number of vertices and the minimum number of edge groups.

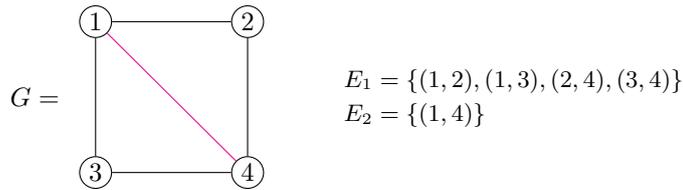
\begin{center}
\begin{tikzpicture}

\tikzstyle{vertex}=[draw, circle, color = black, text opacity = 1, inner sep = 1.4pt]

\node at (-.8,-1) {$G = $};

\node[vertex] (v1) at (0,0) {1};
\node[vertex] (v2) at (2,0) {2};
\node[vertex] (v3) at (0,-2) {3};
\node[vertex] (v4) at (2,-2) {4};

\draw (v1) -- (v2);
\draw (v1) -- (v3);
\draw[magenta] (v1) -- (v4);
\draw (v2) -- (v4);
\draw (v3) -- (v4);

\node at (5.5,-1) {
\begin{tabular}{l}
    $E_1 = \{(1,2), (1,3), (2,4), (3,4)\}$ \\
    $E_2 = \{(1,4)\}$
\end{tabular}
};

\end{tikzpicture}
\end{center}

The maximum proportions of every subgraph $H$ of $G$ are listed in the table below.

\begin{center}
    \begin{tabular}{l|l}
        $H$ & $\MP(H)$  \\ \hline 
        $E_1$ & 1 \\
        $E_2$ & 1 \\
        $E_1 \cup E_2$ & 0.8 \\
    \end{tabular}
\end{center}

Thus, $\min_H \MP(H) = 0.8$, where the minimum is taken over all $H$ that satisfy the conditions of the above theorem. Next, we list the proportions for every group $E_i$ for every possible cut. Additionally, we list the probabilities given to each cut which maximize the minimum expected proportion of the groups.

\begin{center}
    \begin{tabular}{l||l|l|l|l}
        Cut $S$ & $f_S(E_1)/|E_1|$ & $f_S(E_2)/|E_2|$ & $\min_i f_S(E_i)/|E_i|$ & Probability  \\ \hline 
        $\{1\}$ & 0.5 & 1 & 0.5 & 0 \\
        $\{2\}$ & 0.5 & 0 & 0 & 0 \\
        $\{3\}$ & 0.5 & 0 & 0 & 0 \\
        $\{4\}$ & 0.5 & 1 & 0.5 & 2/3 \\
        $\{1, 2\}$ & 0.5 & 1 & 0.5 & 0 \\
        $\{1, 3\}$ & 0.5 & 1 & 0.5 & 0 \\
        $\{1, 4\}$ & 1 & 0 & 0 & 1/3 \\
    \end{tabular}
\end{center}

From this, we can see that $\DFMP(G) = \min \{2/3, 2/3\} = 2/3$, which is less than $\min_H \MP(H) = 0.8$. 
\end{example}

\begin{example}[GW Algorithm Gives Poor Dynamic Fairness Values]
\label{ex:GW_gives_bad_DF}
Consider the same graph $G$ and edge partition $\Gamma$ considered in Example \ref{ex:subgraphBound}; it was shown that $\DFMP(G,\Gamma)=2/3$. Recall that the any feasible solution $x$ to the GW SDP relaxation can be represented by $n$ $n$-dimensional unit vectors, one for each vertex. Using the GW SDP solver in the \texttt{cvxgraphalgs} Python package, we obtain the following solution to the GW SDP relaxation:
$$x_1 = \begin{bmatrix}0\\0\\0\\+1\end{bmatrix}, x_2 = \begin{bmatrix}0\\0\\0\\-1\end{bmatrix},x_3 = \begin{bmatrix}0\\0\\0\\-1\end{bmatrix},x_4 = \begin{bmatrix}0\\0\\0\\+1\end{bmatrix}.$$

% \textcolor{red}{See if there's a gap with single vertex partition--there isn't}

For vertices $u$ and $v$, the angle $\theta_{uv}$ formed between the corresponding vectors $x_u$ and $x_v$ is $\arccos(u \cdot v)$. Goemans and Williamson showed that the probability of an edge $(u,v)$ being cut is given by $\theta_{uv}/\pi$. Since $E_2$ contains just the single edge $(1,4)$, then the proportion of edges cut in $E_2$ is $$\E_{x \sim D} \frac{f_x(E_2)}{|E_2|} = \Pr[(1,4) \text{ is cut}] = \arccos(x_1 \cdot x_4)/\pi = \arccos(1)/\pi = 0.$$
Letting $D$ be the distribution of cuts obtained from the GW algorithm on this example graph, it follows that $\min_{i \in [\ng]} \E_{x \sim D} \frac{f_x(E_i)}{|E_i|} = 0$.
\end{example}

\begin{example}[LP for $\DFMP$ with Edge Utilities]
\label{ex:lp_edge_utilities}
Let $G = (V,E)$ where $V = \{1,2,3,4\}$ and $E = \{(1,2), (2,3), (1,3), (1,4)\}$. Consider the following partition of the edges into $\ng=4$ groups where each group contains exactly one edge:
\[E_1 = \{(1,2)\},
~~E_2 = \{(2,3)\},
~~E_3 = \{(1,3)\},
~~E_4 = \{(1,4)\}.
\]
Expanding the constraints of the above general LP, we obtain the following LP:
\begin{align*}
    \text{max} \ & z \\
    \text{s.t.} \  & z \leq p_{1}+ p_{1,3} + p_{1,4} + p_{1,3,4} + p_{2,3,4} + p_{2,4} + p_{2,3} + p_{2} \\
    & z \leq p_{2}+ p_{1,2} + p_{2,4} + p_{1,2,4} + p_{1,3,4} + p_{3,4} + p_{1,3} + p_{3} \\
    & z \leq p_{1}+ p_{1,2} + p_{1,4} + p_{1,2,4} + p_{2,3,4} + p_{3,4} + p_{2,3} + p_{3} \\
    & z \leq p_{1}+ p_{1,2} + p_{1,3} + p_{1,2,3} + p_{2,3,4} + p_{3,4} + p_{2,4} + p_{4}\\
    & \sum_{S \subseteq V} p_S = 1\\
    & p_S \geq 0 & \forall S \subseteq V. 
\end{align*}

{\it Aside.} Solving the above linear program yields $z^* = 2/3 = \DFMP$, whereas for this graph, we have $\MP = 3/4$.
\end{example}

\begin{proof}[Proof of Proposition \ref{prop:mp-dfmp-edge} ]
Let $G = (V,E)$ be a graph with edge groups $E_1,\ldots,E_\ng$. We begin by proving the upper bound. To that end, let $D^*$ denote the distribution of the naive random cut.\footnote{The naive random cut is the cut $S$ generated by including each vertex independent in $S$ with probability 1/2. See Section~\ref{subsubsec:edge-algorithms} for further discussion.} Then
\[ 
\DFMP = \max_{D \in \mathcal{D}} \min_i \E_{S \sim D} \frac{f_S(E_i)}{|E_i|} \geq \min_i \E_{S \sim D^*} \frac{f_S(E_i)}{|E_i|} = \min_i \frac{1}{2} = \frac{1}{2}.
\]
So, we have that
\[
\MP - \DFMP \leq 1 - \DFMP \leq 1-\frac{1}{2} = \frac{1}{2}.
\]
This proves the upper bound.

We prove tightness using a collection of examples. In particular, consider the graph $G_{k,n}$, which is a copy of $K_{2k}$ with a tail of length $n-2k$ extending from one of its vertices. The edge partition $\Gamma_{k,n}$ for $G_{k,n}$ will consist of two groups: the edges of $K_{2k}$ and the edges of the tail. For example, $G_{2,n}$ is shown below.

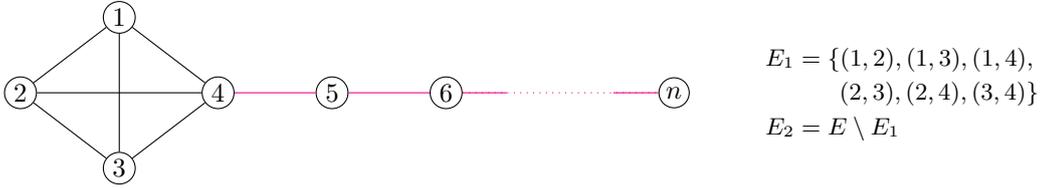
\begin{center}
\begin{tikzpicture}

\tikzstyle{vertex}=[draw, circle, color = black, text opacity = 1, inner sep = 1.4pt]

\node[vertex] (v1) at (-1.3,1) {1};
\node[vertex] (v3) at (-1.3,-1) {3};
\node[vertex] (v2) at (-2.6,0) {2};
\node[vertex] (v4) at (0,0) {4};
\node[vertex] (v5) at (1.5,0) {5};
\node[vertex] (v6) at (3,0) {6};
\node[vertex] (vn) at (6,0) {$n$};

\draw (v1) -- (v2);
\draw (v1) -- (v3);
\draw (v1) -- (v4);
\draw (v2) -- (v3);
\draw (v2) -- (v4);
\draw (v3) -- (v4);
\draw[magenta] (v4) -- (v5);
\draw[magenta] (v5) -- (v6);
\draw[magenta, dotted] (v6) -- (vn);

\draw[magenta] (v6) -- (3.8,0);
\draw[magenta] (5.2,0) -- (vn);

\node at (9,0) {
\begin{tabular}{l}
    $E_1 = \{(1,2), (1,3), (1,4),$ \\
    $~~~~~~~~~(2,3), (2,4), (3,4)\}$ \\
    $E_2 = E \setminus E_1$
\end{tabular}
};

\end{tikzpicture}
\end{center}

Note that the number of nodes in $G_{k,n}$ is $n$ and the number of edges is $n-2k + \binom{2k}{2}$. In this case, 
\[
\MP(G_{k,n}) = \frac{k^2 + n-2k}{\binom{2k}{2} + n-2k} \xrightarrow{~n \to \infty~} 1,
\]
and $\DFMP(G_{k,n}, \Gamma_{k,n}) = \frac{k^2}{\binom{2k}{2}} = \frac{k}{2k-1}$. Thus, for a fixed $k$,
\[
\MP(G_{k,n}) - \DFMP(G_{k,n}, \Gamma_{k,n}) \xrightarrow{~n \to \infty~} 1 - \frac{k}{2k-1} = \frac{k-1}{2k-1}.
\]
Thus, by choosing $k$ to be large enough, one can obtain a graph $G$ and edge partition $\Gamma$ for which $\MP(G) - \DFMP(G, \Gamma)$ is arbitrarily close to 1/2. 
\end{proof}

\begin{proof}[Proof of Proposition \ref{prop:dfmp-sfmp-edge}]
Consider the $(2n+1)$-cycle $C_{2n+1}$ with $2n+1$ singleton edge sets 
\[
E_1=\{(0,1)\}, ~E_2= \{(1, 2)\}, \ldots, ~E_{2n+1}=\{(2n, 0)\}.
\]
This graph is depicted below for $n=2$.

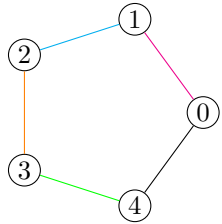
\begin{center}
\begin{tikzpicture}

\tikzstyle{vertex}=[draw, circle, color = black, fill = white, text opacity = 1, inner sep = 1.4pt]

\foreach \s in {0,...,4}{
    % \node[vertex] (v\s) at (2*rand,2*rand){};
    \node[vertex] (v\s) at ({360/5 * \s}:1.3cm){\s};
}

\draw[magenta] (v0) -- (v1);
\draw[cyan] (v1) -- (v2);
\draw[orange] (v2) -- (v3);
\draw[green] (v3) -- (v4);
\draw (v4) -- (v0);

\end{tikzpicture}
\end{center}

In this case, since $2n+1$ is odd, $\SFMP(C_{2n+1}) = 0$. Now suppose that the cut which cuts all edges except $(k,k+1 \pmod{2n+1})$ is chosen with probability $1/(2n+1)$. In this case,
\[
\E_S \frac{f_S(E_k)}{|E_k|} = 1-\frac{1}{2n+1} ~~\mbox{for}~k=1,\ldots,2n+1.
\]
It follows that $\DFMP(C_{2n+1}) \geq 1-1/(2n+1)$. Letting $n \to \infty$, we see that $\DFMP(C_{2n+1}) - \SFMP(C_{2n+1}) \to 1$.
\end{proof}

\section{Proofs from Section~\ref{sec:nodeUtilities}}
\label{sec:proofs_nodeUtilities}

% \begin{proposition} \label{prop:mp-dfmp-node}
% For any $\varepsilon > 0$, there exists a graph $G$ and node partition $\Gamma$ for which $\MP(G) - \DFMP(G,\Gamma) > 1-\varepsilon$.
% \end{proposition}

% The proof of this proposition follows from the observation that $\DFMP(G) \leq \min_i \frac{\Delta(V_i)}{\Delta(G)}$ and can be found in Appendix~\ref{app:nodes}.

% Below is the proo

\begin{proof}[Proof of Proposition \ref{prop:mp-dfmp-node}]
Consider the graph $G_{k,r}$ ($G_{2,3}$ is depicted below) which is a $2k$-cycle connected by a single edge to a $K_{r,r}$. Let the partition be $\Gamma_{k,r} = \{V_1,V_2\}$, where nodes of the cycle comprise $V_1$, and the nodes of the complete bipartite graph comprise $V_2$.

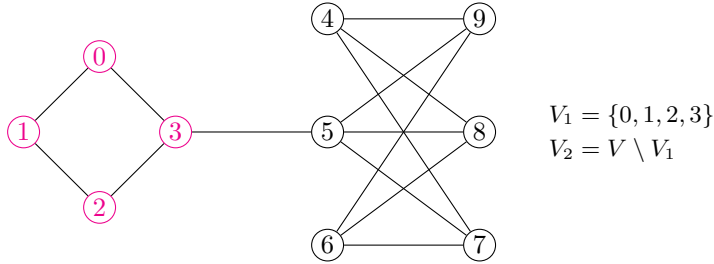
\begin{center}
\begin{tikzpicture}

\tikzstyle{vertex}=[draw, circle, color = black, text opacity = 1, inner sep = 1.4pt]

\node[vertex, magenta] (v0) at (-1,1) {0};
\node[vertex, magenta] (v1) at (-2,0) {1};
\node[vertex, magenta] (v2) at (-1,-1) {2};
\node[vertex, magenta] (v3) at (0,0) {3};
\node[vertex] (v4) at (2,1.5) {4};
\node[vertex] (v5) at (2,0) {5};
\node[vertex] (v6) at (2,-1.5) {6};
\node[vertex] (v7) at (4,-1.5) {7};
\node[vertex] (v8) at (4,0) {8};
\node[vertex] (v9) at (4,1.5) {9};

\draw (v0) -- (v1);
\draw (v1) -- (v2);
\draw (v2) -- (v3);
\draw (v3) -- (v0);
\draw (v3) -- (v5);
\draw (v4) -- (v7);
\draw (v4) -- (v8);
\draw (v4) -- (v9);
\draw (v5) -- (v7);
\draw (v5) -- (v8);
\draw (v5) -- (v9);
\draw (v6) -- (v7);
\draw (v6) -- (v8);
\draw (v6) -- (v9);

\node at (6,0) {
\begin{tabular}{l}
    $V_1 = \{0, 1, 2, 3\}$ \\
    $V_2 = V \setminus V_1$
\end{tabular}
};

\end{tikzpicture}
\end{center}

In this case, for any cut $S$, we have that
\[
\frac{f_S(V_1)}{|V_1|} \leq \frac{1}{|V_1|} \sum_{u \in V_1} \frac{\deg u}{\Delta(G_{k,r})} \leq \frac{1}{2k} \sum_{u \in V_1} \frac{3}{r+1} = \frac{3}{r+1} \xrightarrow{r \to \infty} 0.
\]
It follows that $\DFMP(G_{k,r},\Gamma_{k,r}) \to 0$ as $r \to \infty$. At the same time, we have that
\begin{align*}
    \MP(G_{k,r}) = \frac{2 \MC(G_{k,r})}{|V| \Delta(G_{k,r})} = \frac{2(r^2 + 2k + 1)}{(2r + 2k)(r+1)} \xrightarrow{r \to \infty} 1.
\end{align*}
By choosing $r$ large enough, we can ensure that $\MP(G_{k,r}) - \DFMP(G_{k,r},\Gamma_{k,r}) > 1-\varepsilon$.
\end{proof}

\begin{proof}[Proof of Proposition~\ref{prop:dfmp-minus-sfmp-node}]
We begin with the first statement. Consider the $(2n+1)$-cycle $C_{2n+1}$ with a node partition $\Gamma_{2n+1}$ containing $2n+1$ singleton vertex sets $V_i = \{i-1\}$ for $i=1,\ldots,2n+1$. This graph is depicted below for $n=2$.

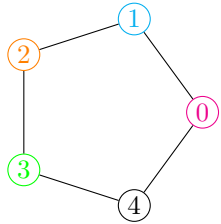
\begin{center}
\begin{tikzpicture}

\tikzstyle{vertex}=[draw, circle, color = black, text opacity = 1, inner sep = 1.4pt]

\node[vertex, magenta] (v0) at ({360/5 * 0}:1.3cm){0};
\node[vertex, cyan] (v1) at ({360/5 * 1}:1.3cm){1};
\node[vertex, orange] (v2) at ({360/5 * 2}:1.3cm){2};
\node[vertex, green] (v3) at ({360/5 * 3}:1.3cm){3};
\node[vertex] (v4) at ({360/5 * 4}:1.3cm){4};

\draw (v0) -- (v1);
\draw (v1) -- (v2);
\draw (v2) -- (v3);
\draw (v3) -- (v4);
\draw (v4) -- (v0);

\end{tikzpicture}
\end{center}

In this case, since $2n+1$ is odd, $\SFMP(C_{2n+1},\Gamma_{2n+1}) = 1/2$. Now suppose that the cut which cuts all edges except $(k,k+1 \pmod{2n+1})$ is chosen with probability $1/(2n+1)$. In this case,
\[
\E_S f_S(V_k) = 1 \cdot \bigg(1-\frac{2}{2n+1}\bigg) + \frac{1}{2} \cdot \bigg( \frac{2}{2n+1}\bigg)~~\mbox{for}~k=1,\ldots,2n+1.
\]
It follows that $\DFMP(C_{2n+1},\Gamma_{2n+1}) \geq 1-2/(2n+1)$. Letting $n \to \infty$, we see that $\DFMP(C_{2n+1},\Gamma_{2n+1}) - \SFMP(C_{2n+1},\Gamma_{2n+1}) \to 1/2$.

Next, we show tightness. Observe that for any graph $G$, any vertex partition $\Gamma=\{V_1,\ldots,V_m\}$, and any cut $S$, we have that
\[
\frac{f_S(V_i)}{|V_i|}  \leq \frac{1}{|V_i|} \sum_{u \in V_i} \frac{\deg(u)}{\Delta(G)} ~~\mbox{for all $i \in [m]$}.
\]
It follows that 
\begin{equation} \label{eqn:dfmp-upper-bound}
\DFMP(G,\Gamma) \leq \min_i \frac{1}{|V_i|} \sum_{u \in V_i} \frac{\deg(u)}{\Delta(G)}.
\end{equation}
To lower bound $\SFMP(G,\Gamma)$, let $S$ be a cut resulting from the local search algorithm. In this case, the number of edges incident to $u$ that cross the cut $S$ is at least $\deg(u)/2$ for all vertices $u \in V$. It follows that
\[
\frac{f_S(V_i)}{|V_i|}  \geq \frac{1}{|V_i|} \sum_{u \in V_i} \frac{\deg(u)}{2\Delta(G)} ~~\mbox{for all $i \in [m]$}.
\]
So,
\begin{equation} \label{eqn:sfmp-lower-bound}
\SFMP(G,\Gamma ) \geq \min_i \frac{1}{|V_i|} \sum_{u \in V_i} \frac{\deg(u)}{2\Delta(G)}.
\end{equation}
In summation, we have that
\begin{align*}
\DFMP(G,\Gamma) - \SFMP(G,\Gamma) &\leq  %\min_i \frac{1}{|V_i|} \sum_{u \in V_i} \frac{\deg(u)}{\Delta(G)} - \min_i \frac{1}{|V_i|} \sum_{u \in V_i} \frac{\deg(u)}{2\Delta(G)} &\mbox{by (\ref{eqn:dfmp-upper-bound}) and (\ref{eqn:sfmp-lower-bound})} \\&= 
\frac{1}{2} \bigg( \min_i \frac{1}{|V_i|} \sum_{u \in V_i} \frac{\deg(u)}{\Delta(G)} \bigg) &\mbox{by (\ref{eqn:dfmp-upper-bound}) and (\ref{eqn:sfmp-lower-bound})} \\
&\leq \frac{1}{2}.
\end{align*}
\end{proof}

\end{document}